\documentclass[12pt]{amsart}  

\usepackage{graphicx}
\usepackage{here} 
\usepackage{array} 

\usepackage{vmargin}
\usepackage{amssymb}
\usepackage{mathrsfs}
\usepackage[all]{xy}
\usepackage[usenames,dvipsnames]{color}
\usepackage{soul}
\RequirePackage[colorlinks,linkcolor=blue,citecolor=LimeGreen,urlcolor=red]{hyperref} 
\usepackage{amsmath}

\newtheorem{theorem}{Theorem}[section]

\newtheorem{cor}[theorem]{Corollary}

\newtheorem{defi}[theorem]{Definition}

\newtheorem{lemma}[theorem]{Lemma}

\newtheorem{prop}[theorem]{Proposition}
\newtheorem{remark}[theorem]{Remark}

\numberwithin{equation}{section}

\newcommand{\R}{\mathbb{R}}

\newcommand{\N}{\mathbb{N}}

\newcommand{\T}{\mathbb{T}}

\newcommand{\abs}[1]{\left|#1\right|}
\newcommand{\eps}{\varepsilon}
\newcommand{\norm}[1]{\left\|#1\right\|}

\renewcommand{\leq}{\leqslant}
\renewcommand{\geq}{\geqslant}

\renewcommand{\bar}{\overline}

\renewcommand{\tilde}{\widetilde}

\newcommand\restr[2]{{
  \left.\kern-\nulldelimiterspace 
  #1 
  \right|_{#2} 
  }}

\def\signmb{\bigskip \begin{center} {\sc
Marc Briant\par\vspace{3mm}
Brown University\par
Division of Applied Mathematics\par
182 George Street, Box F
Providence, RI 02192, USA\par
\vspace{3mm}
e-mail:} \tt{briant.maths@gmail.com} \end{center}}

\begin{document} 

\title[EXPONENTIAL LOWER BOUND FOR BOLTZMANN WITH MAXWELLIAN DIFFUSION]{INSTANTANEOUS EXPONENTIAL LOWER BOUND FOR SOLUTIONS TO THE BOLTZMANN EQUATION WITH MAXWELLIAN DIFFUSION BOUNDARY CONDITIONS}
\author{M. Briant}

\maketitle

\begin{abstract}
We prove the immediate appearance of an exponential lower bound, uniform in time and space, for continuous mild solutions to the full Boltzmann equation in a $C^2$ convex bounded domain with the physical Maxwellian diffusion boundary conditions, under the sole assumption of regularity of the solution. We investigate a wide range of collision kernels, with and without Grad's angular cutoff assumption. In particular, the lower bound is proven to be Maxwellian in the case of cutoff collision kernels. Moreover, these results are entirely constructive if the initial distribution contains no vacuum, with explicit constants depending only on the \textit{a priori} bounds on the solution. 
\end{abstract}

\vspace*{10mm}



\smallskip
\textbf{Acknowledgements:} I would like to acknowledge Fran\c cois Golse for suggesting me to look at this problem.
\\ This work was supported by the $150^{th}$ Anniversary Postdoctoral Mobility Grant of the London Mathematical Society.
\tableofcontents

\section{Introduction} \label{sec:intro}

The Boltzmann equation rules the dynamics of rarefied gas particles moving in a domain $\Omega$ of $\R^d$ with velocities in $\R^d$ ($d \geq 2$) when the only interactions taken into account are elastic binary collisions. More precisely, the Boltzmann equation describes the time evolution of $f(t,x,v)$, the distribution of particles in position and velocity, starting from an initial distribution $f_0(x,v)$ .
\par In the present article we are interested in the case where the gas stays inside a domain of which walls are heated at a constant temperature $T_\partial$. Contrary to the classical specular (billiard balls) or bounce-back reflections boudary conditions, the temperature of the boundary generates a diffusion towards the inside of the domain which prevents the usual preservation of energy of the gas.
\par We investigate the case where $\Omega$ is a $C^2$ convex bounded domain and that the boundary conditions are Maxwellian diffusion. The Boltzmann equation reads

\begin{eqnarray}
\forall t \geq 0 &,& \:\forall (x,v) \in \Omega \times \R^d,\quad  \partial_t f + v\cdot \nabla_x f = Q(f,f),\label{BE}
\\ && \:\forall (x,v) \in \bar{\Omega} \times \R^d,\quad f(0,x,v) = f_0(x,v), \nonumber
\end{eqnarray}
with $f$ satisfying the Maxwellian diffusion boundary condition:
$$\forall (t,x,v) \in \R^{*+}\times\partial \Omega \times \R^d, \quad f(t,x,v)= f_{\partial}(t,x,v),$$
where
\begin{equation}\label{fpartial}
f_\partial(t,x,v)= \left[\int_{v\cdot n(x)>0} f(t,x,v)\left(v\cdot n(x)\right)\:dv\right]\frac{1}{\left(2\pi\right)^{\frac{d-1}{2}}T_{\partial}^{\frac{d+1}{2}}}e^{-\frac{\abs{v}^2}{2T_{\partial}}},
\end{equation}
with $n(x)$ denoting the outwards normal to $\bar{\Omega}$ at $x$ on $\partial\Omega$. This boundary condition expresses the physical process where particles are absorbed by the wall and then emitted back into $\Omega$ according to the thermodynamical equilibrium distribution between the wall and the gas.

\bigskip
The operator $Q(f,f)$ encodes the physical properties of the interactions between two particles. This operator is quadratic and local in time and space. It is given by 
$$Q(f,f) =  \int_{\R^d\times \mathbb{S}^{d-1}}B\left(|v - v_*|,\mbox{cos}\:\theta\right)\left[f'f'_* - ff_*\right]dv_*d\sigma,$$
where $f'$, $f_*$, $f'_*$ and $f$ are the values taken by $f$ at $v'$, $v_*$, $v'_*$ and $v$ respectively. Define:
$$\left\{ \begin{array}{rl} \displaystyle{v'} & \displaystyle{= \frac{v+v_*}{2} +  \frac{|v-v_*|}{2}\sigma} \vspace{2mm} \\ \vspace{2mm} \displaystyle{v' _*}&\displaystyle{= \frac{v+v_*}{2}  -  \frac{|v-v_*|}{2}\sigma} \end{array}\right., \: \mbox{and} \quad \mbox{cos}\:\theta = \langle \frac{v-v_*}{\abs{v-v_*}},\sigma\rangle .$$
We recognise here the conservation of kinetic energy and momentum when two particles of velocities $v$ and $v_*$ collide to give two particles of velocities $v'$ and $v'_*$.
\bigskip
The collision kernel $B \geq 0$ contains all the information about the interaction between two particles and is determined by physics (see \cite{Ce} or \cite{Ce1} for a formal derivation for the hard sphere model of particles). In this paper we shall only be interested in the case of $B$ satisfying the following product form
\begin{equation}\label{assumptionB}
B\left(|v - v_*|,\mbox{cos}\:\theta\right) = \Phi\left(|v - v_*|\right)b\left( \mbox{cos}\:\theta\right),
\end{equation}
which is a common assumption as it is more convenient and also covers a wide range of physical applications.
Moreover, we shall assume that $\Phi$ satisfies either
\begin{equation}\label{assumptionPhi}
\forall z \in \R,\quad c_\Phi \abs{z}^\gamma \leq \Phi(z) \leq C_\Phi \abs{z}^\gamma
\end{equation}
or a mollified assumption
\begin{equation}\label{assumptionPhimol}
\left\{\begin{array}{rl}\displaystyle{\forall \abs{z} \geq 1 \in \R,}&\displaystyle{\quad c_\Phi \abs{z}^\gamma \leq \Phi(z) \leq C_\Phi \abs{z}^\gamma} \vspace{2mm} \\\vspace{2mm} \displaystyle{\forall \abs{z} \leq 1 \in \R,}&\displaystyle{\quad c_\Phi \leq \Phi(z) \leq C_\Phi,} \end{array}\right.
\end{equation}
$c_\Phi$ and $C_\Phi$ being strictly positive constants and $\gamma$ in $(-d,1]$. The collision kernel is said to be ``hard potential" in the case of $\gamma >0$, ``soft potential" if  $\gamma < 0$ and ``Maxwellian" if $\gamma = 0$.
\par Finally, we shall consider $b$ to be a continuous function on $\theta$ in $(0,\pi]$, strictly positive near $\theta \sim \pi/2$, which satisfies
\begin{equation}\label{assumptionb}
b\left(\mbox{cos}\:\theta\right)\mbox{sin}^{d-2}\theta\mathop{\sim}\limits_{\theta \to 0^+}b_0 \theta^{-(1+\nu)}
\end{equation}
for $b_0 >0$ and $\nu$ in $(-\infty,2)$. The case when $b$ is locally integrable, $\nu < 0$, is referred to by the Grad's cutoff assumption (first introduce in \cite{Gr1}) and therefore $B$ will be said to be a cutoff collision kernel. The case $\nu \geq 0$ will be designated by non-cutoff collision kernel.
\bigskip


\subsection{Motivations and comparison with previous results} \label{subsec:previously}
The aim of this article is to show and to quantify the strict positivity of the solutions to the Boltzmann equation when  the gas particles move in a domain with boundary conditions. In that sense, it continues the study started in \cite{Bri2} about exponential lower bounds for solutions to the Boltzmann equation when specular refletions boundary conditions were taken into account.
\par More precisely, we shall prove that continuous solutions to the Boltzmann equation, with Maxwellian diffusion boundary conditions in a $C^2$ convex bounded domain, which have uniformly bounded energy satisfy an immediate exponential lower bound:
$$\forall t \geq t_0, \forall (x,v) \in \T^d\times\R^d, \quad f(t,v) \geq C_1 e^{-C_2\abs{v}^{K}},$$
for all $t_0>0$. Moreover, in the case of collision kernel with angular cutoff we recover a Maxwellian lower bound
$$\forall \tau>0, \: \exists \rho,\theta >0, \: \forall t \geq \tau,\:\forall (x,v) \in \Omega\times\R^d, \quad f(t,x,v) \geq \frac{\rho}{(2\pi\theta)^{d/2}}e^{-\frac{\abs{v}^2}{2\theta}}.$$
\par We would like to emphasize that, in the spirit of \cite{Bri2}, our results show that the gas will instantaneously fill up the whole domain even if the initial configuration contains vacuum. Indeed, they only require some regularity on the solution and no further assumption on its local density. Previous studies assumed the latter to be uniformly bounded from below, which is equivalent of assuming \textit{a priori} either that there is no vacuum or that the solution is strictly positive.
\par Moreover, the present results only require solutions to the Boltzmann equation to be continuous away from the grazing set
\begin{equation}\label{grazingset}
\Lambda_0 = \left\{(x,v)\in \partial\Omega\times\R^d,\:\quad n(x)\cdot v =0\right\},
\end{equation}
which is a property that is known to hold in the case of Maxwellian diffusion boundary conditions \cite{Gu6}.

\bigskip
The issue of quantifying the positivity of solutions has been investigated for a long time since it not only presents a great physical interest but also appears to be of significant importance for the mathematical study of the Boltzmann equation. Indeed, exponential lower bounds are essential for entropy-entropy production methods used to describe long-time behaviour for kinetic equations \cite{DV}\cite{DV1}. More recently, such lower bounds were needed to prove uniqueness of solutions to the Boltzmann equation in $L^1_vL^\infty_x\left(1+\abs{v}^{2+0}\right)$ \cite{GMM}.
\par Several works quantified the study of an explicit lower bound for solutions to the Boltzmann equation. We give here a brief overview and we refer the interested reader to the more detailed description in \cite{Bri2}.
\par The first result about the strict positivity of solutions to the Boltzmann equation was derived by Carleman \cite{Ca}. Noticing that a part $Q^+$ of the Boltzmann operator $Q$ satisfies a spreading property, roughly speaking
$$Q^+(\mathbf{1}_{B(\bar{v},r)},\mathbf{1}_{B(\bar{v},r)}) \geq C_+\mathbf{1}_{B\left(\bar{v},\sqrt{2}r\right)},$$
with $C_+<1$ (see Lemma $\ref{lem:Q+spread}$ for an exact statement), he proved the immediate creation of an exponential lower bound for a certain class of solutions (radially symmetric in velocity) to the spatially homogeneous equation with hard potential kernel with angular cutoff. The latter result was improved to a Maxwellian lower bound and extended to the case of non-radially symmetric solutions to the spatially homogeneous equation with hard potential and cutoff by Pulvirenti and Wennberg \cite{PW}.
\par Finally, the study in the case of the full equation has been tackled by Mouhot \cite{Mo2} in the case of the torus $\Omega=\T^d$, and more recently by the author \cite{Bri2} in $C^2$ convex bounded domains with specular boundary conditions in all dimension $d$. In both articles, a Maxwellian lower bound is derived for solutions to the full Boltzmann equation with angular cutoff (both with hard of soft potentials) and they showed the appearance of an exponential lower bound in the non-cutoff case.

\bigskip
Our present results show that the previous properties proven for the full Boltzmann equation \cite{Mo2}\cite{Bri2} still hold in the physically relevant case of the Maxwellian diffusion generated by the boundary of the domain $\Omega$. This is done for all the physically relevant collision kernels such as with and without angular cutoff and hard and soft potentials. Moreover, in the case of a solution that has uniformly bounded local mass and entropy, the proofs are entirely constructive and the constants are explicit and only depend on the \textit{a priori} bounds on the solution and the geometry of the domain, which is of great physical interest for the study of the spreading of gas into a domain with heated walls.
\par There are two key contributions in the present article. The main one is a quantification of the strict positivity of the Maxwellian diffusion process thanks to a combination of a localised positivity of the solution and a geometrical study of the rebounds against a convex boundary. Roughly speaking, we show that the wall instantaneously starts to diffuse in all directions and that its diffusion is uniformly bounded from below. The second one is a spreading method combining the effects used in previous studies and the exponential diffusive process (see next section for details).
\bigskip


\subsection{Our strategy} \label{subsec:strategy}

\par The main strategy to tackle this result relies on the breakthrough of Carleman \cite{Ca}, namely finding an ``upheaval point'' (a first minoration uniform in time but localised in velocity) and spreading this bound, thanks to the spreading property of the $Q^+$ operator, in order to include larger and larger velocities and finally compare it to an exponential. The case of the spatially non-homogeneous equation \cite{Mo2}\cite{Bri2} was dealt by finding a spreading method that was invariant along the flow of characteristics.

\bigskip
The creation of ``upheaval points" (localised in space because of boundary effects) is essantially the method developed in \cite{Bri2} for general continuous initial datum or the one of \cite{PW}\cite{Mo2} for constructive purposes. There is a new technical treatment of the time of appearance of such lower bounds but does not present any new difficulties.
\par The main issue is the use of a spreading method that would be invariant along the flow of characteristics \cite{Mo2}\cite{Bri2}. In the case of Maxwellian diffusion boundary conditions  characteristic trajectories are no longer well defined. Indeed, once a trajectory touches the boundary it is absorbed an re-emitted in all the directions. Characteristic trajectories are therefore only defined in between two consecutive rebounds and one cannot hope to use the invariant arguments developed in \cite{Mo2}\cite{Bri2}.
\par The case of the torus, studied in \cite{Mo2}, indicates that without boundaries the exponential lower bound is created after time $t=0$ as quickly as one wants. In the case of a bounded domain with specular reflection boundary conditions \cite{Bri2}, this minoration also occurs immediately. Together, it roughly means that one can expect to obtain an exponential lower bound on straight lines uniformly on how close the particle is from the boundary. Therefore we expect  the same kind of uniform bounds to arise on each characteristic trajectory in between two consecutive rebounds as long as the Maxwellian diffusion emitted by the boundary is uniformly non-negative.
\par Our strategy is therefore to first prove that the boundary condition produces a strictly positive quantity uniformly towards the interior of $\Omega$ and then to find a method to spread either this diffusion or the localised ``upheaval points". More precisely, if there is no contact during a time $\tau>0$ we expect to be able to use the spreading method developed in \cite{Mo2} from the initial lower bounds. Else there is a contact during the interval $[0,\tau]$ we cannot hope to use the latter spreading method, nor its more general characteristics invariant version derived in \cite{Bri2}, since the Maxwellian diffusion boundary condition acts like an absorption for particles arriving on the boundary. But this boundary condition also diffuses towards the interior of the domain $\Omega$ what it has absorbed. This diffusion follows an exponential law and therefore of the shape of a Maxwellian lower bound that we manage to keep along the incoming characteristic trajectory.

\bigskip
Collision kernels satisfying a cutoff property as well as collision kernels with a non-cutoff property will be treated following the strategy described above. The only difference is the decomposition of the Boltzmann bilinear operator $Q$ we consider in each case. In the case of a non-cutoff collision kernel, we shall divide it into a cutoff collision kernel and a remainder. The cutoff part will already be dealt with and a careful control of the $L^\infty$-norm of the remainder will give us the expected exponential lower bound, which decreases faster than a Maxwellian.
\bigskip


\subsection{Organisation of the paper}\label{subsec:organization}

Section $\ref{sec:mainresults}$ is dedicated to the statement and the description of the main results proven in this article. It contains three different parts
\par Section $\ref{subsec:notations}$ defines all the notations which will be used throughout the article.
\par The last subsections, $\ref{subsec:maincutoff}$ and $\ref{subsec:mainnoncutoff}$, are dedicated to a mathematical formulation of the results related to the lower bound in, respectively, the cutoff case and the non-cutoff case, described above. It also defines the concept of mild solutions to the Boltzmann equation in each case.

\bigskip
Section $\ref{sec:cutoffcase}$ deals with the case of the immediate Maxwellian lower bound for collision kernels with angular cutoff. As described in our strategy, it follows three steps.
\par Section $\ref{subsec:1stlocalisedbounds}$ generates the localised ``upheaval points" for general initial datum. A constructive approach to that problem is given in Section $\ref{subsec:constructiveupheavalpoint}$.
\par The uniform positivity of the Maxwellian diffusion is proven in Section $\ref{subsec:diffusioneffects}$.
\par Finally, Section $\ref{subsec:maxwellboundcutoff}$ combines the standard spreading methods with the diffusion process to prove the expected instantaneous Maxwellian lower bound.

\bigskip
To conclude, Section $\ref{sec:noncutoff}$ proves the immediate appearance of an exponential lower bound in the case of collision kernels without angular cutoff. We emphasize here that this Section actually explains the adaptations required compared to the case of collision kernels with the cutoff property.

\bigskip

\section{Main results} \label{sec:mainresults}

We begin with the notations we shall use all along this article.

\subsection{Notations}\label{subsec:notations}

First of all, we denote $\langle \cdot \rangle = \sqrt{1 + \abs{\cdot}^2}$ and $y^+ = \max\{0,y\}$, the positive part of $y$.
\par\textbf{Functional spaces.}
This study will hold in specific functional spaces regarding the $v$ variable that we describe here and use throughout the sequel. For all $p$ in $[1,\infty]$, we use the shorthand notation for Lebesgue spaces $L^p_v = L^p\left(\R^d\right)$.
\par For $p \in [1,\infty]$ and $k \in \N$ we use the Sobolev spaces $W^{k,p}_v$ by the norm $$\norm{f}_{W^{k,p}_v} = \left[\sum\limits_{\abs{s}\leq k}\norm{\partial^sf(v)}^p_{L^p_v}\right]^{1/p}.$$


\bigskip
\textbf{Physical observables and hypotheses.}
In the sequel of this study, we are going to need bounds on some physical observables of solution to the Boltzmann equation $\eqref{BE}$.
\par We consider here a function $f(t,x,v) \geq 0$ defined on $[0,T)\times \Omega\times \R^d$ and we recall the definitions of its local hydrodynamical quantities.

\bigskip
\begin{itemize}
\item its local energy  $$e_f (t,x) = \int_{\R^d}\abs{v}^2f(t,x,v)dv,$$
\item its local weighted energy  $$e'_f (t,x) = \int_{\R^d}\abs{v}^{\tilde{\gamma}}f(t,x,v)dv,$$ where $\tilde{\gamma} = (2+\gamma)^+$,
\item its local $L^p$ norm ($p \in [1,+\infty)$) $$l^p_f(t,x) = \norm{f(t,x,\cdot)}_{L^p_v},$$
\item its local $W^{2,\infty}$ norm $$w_f(t,x) = \norm{f(t,x,\cdot)}_{W^{2,\infty}_v}.$$
\end{itemize}
\bigskip

Our results depend on uniform bounds on those quantities and therefore, to shorten calculations we will use the following

\begin{eqnarray*}
E_f =  \sup\limits_{(t,x)\in [0,T)\times \Omega}e_f(t,x) &,& E'_f = \sup\limits_{(t,x)\in [0,T)\times \Omega}e'_f(t,x) ,
\\ L^p_f = \sup\limits_{(t,x)\in [0,T)\times \Omega} l^p_f(t,x) &,& W_f = \sup\limits_{(t,x)\in [0,T)\times \Omega} w_f(t,x).
\end{eqnarray*}

\bigskip
In our theorems we are giving a priori lower bound results for solutions to $\eqref{BE}$ satisfying some properties about their local hydrodynamical quantities. Those properties will differ depending on which case of collision kernel we are considering. We will take them as assumptions in our proofs and they are the following.
\begin{itemize}
\item In the case of hard or Maxwellian potentials with cutoff ($\gamma\geq 0$ and $\nu <0$):
\begin{equation}\label{assumption1}
E_f < +\infty.
\end{equation}
\item In the case of a singularity of the kinetic collision kernel ($\gamma \in (-d,0)$) we shall make the additional assumption
\begin{equation}\label{assumption2}
L^{p_\gamma}_f < +\infty,
\end{equation}
where $p_\gamma  > d/(d+\gamma)$.
\item In the case of a singularity of the angular collision kernel ($\nu \in [0,2)$) we shall make the additional assumption
\begin{equation}\label{assumption3}
W_f < +\infty, \:\: E'_f < + \infty.
\end{equation}
\end{itemize}

Assumption $\eqref{assumption2}$ implies the boundedness of the local entropy and if $\gamma \leq 0$ we have $E'_f \leq E_f$ and so in some cases several assumptions might be redundant.
\par Moreover, in the case of the torus with periodic conditions or the case of bounded domain with specular boundary reflections \cite{Bri2}, solutions to $\eqref{BE}$ also satisfied the conservation of the total mass and the total energy. The case with Maxwellian diffusion boundary conditions only preserves, in general (see \cite{Vi2} for instance), the total mass:

\begin{equation}\label{massconservation}
\exists \mbox{M} \geq 0,\:\forall t \in \R^+, \quad \int_{\Omega}\int_{\R^d} f(t,x,v)\:dxdv = M.
\end{equation}


\bigskip
\textbf{Characteristic trajectories.}
The characteristic trajectories of the equation are only defined between two consecutive rebounds against the boundary and they are given by straight lines that we will denote by
\begin{equation}\label{characteristics}
\forall\: 0\leq s\leq t,\:\forall (x,v)\in \R^d\times\R^d,\quad X_{s,t}(x,v) = x-(t-s)v.
\end{equation}
Because $\bar{\Omega}$ is a closed set, we can define the first time of contact between a backward trajectory and $\partial\Omega$:
\begin{equation}\label{tpartial}
\forall (x,v) \in \bar{\Omega}\times\R^d,\quad t_\partial(x,v) = \min\left\{t \geq 0: \: x-vt \in \partial\Omega\right\},
\end{equation}
and the contact point between such a trajectory and the boundary:
\begin{equation}\label{xpartial}
\forall (x,v) \in \bar{\Omega}\times\R^d,\quad x_\partial(x,v) = x-vt_\partial(x,v)
\end{equation}
\bigskip

 
\subsection{Maxwellian lower bound for cutoff collision kernels}\label{subsec:maincutoff}

The final theorem we prove in the case of cutoff collision kernel is the immediate appearance of a uniform Maxwellian lower bound. We use, in that case, the Grad's splitting for the bilinear operator $Q$ such that the Boltzmann equation reads
\begin{eqnarray*}
Q(g,h) &=&  \int_{\R^d\times \mathbb{S}^{d-1}}\Phi\left(|v - v_*|\right)b\left( \mbox{cos}\theta\right)\left[h'g'_* - hg_*\right]dv_*d\sigma
\\     &=& Q^+(g,h) - Q^-(g,h),
\end{eqnarray*}
where we used the following definitions
\begin{eqnarray}
Q^+(g,h) &=& \int_{\R^d\times \mathbb{S}^{d-1}}\Phi\left(|v - v_*|\right)b\left( \mbox{cos}\theta\right)h'g'_*\: dv_*d\sigma,\nonumber
\\ Q^-(g,h) &=& n_b \left(\Phi * g (v)\right)h = L[g](v)h, \label{gradsplitting}
\end{eqnarray}
where
\begin{equation}\label{nb}
n_b = \int_{\mathbb{S}^{d-1}}b\left(\mbox{cos}\:\theta\right)d\sigma = \left|\mathbb{S}^{d-2}\right|\int_0^\pi b\left(\mbox{cos}\:\theta\right) \mbox{sin}^{d-2}\theta \:d\theta.
\end{equation}

\bigskip
As already mentionned, the characteristics of our problem can only be defined in between two consecutive rebounds against $\partial\Omega$. We can therefore define a mild solution of the Boltzmann equation in the cutoff case, which is expressed by a Duhamel formula along the characteristics. This weaker form of solutions is actually the key point for our result and also gives a more general statement.

\bigskip
\begin{defi}\label{def:mildcutoff}
Let $f_0$ be a measurable function non-negative almost everywhere on $\bar{\Omega}\times \R^d$.
\\ A measurable function $f = f(t, x, v)$ on $[0, T)\times\bar{\Omega}\times \R^d$ is a mild solution of the Boltzmann equation associated to the initial datum $f_0(x, v)$ if 
\begin{enumerate}
\item $f$ is non-negative on $\bar{\Omega}\times \R^d$,
\item for every $(t, x, v)$ in $\R^+\times \Omega\times \R^d$:
$$s \longmapsto L[f(t,X_{s,t}(x,v),\cdot)](v), \quad t \longmapsto Q^+[f(t,X_{s,t}(x,v),\cdot), f(t,X_{s,t}(x,v),\cdot)](v)$$
are in $L^1_{loc}([0,T))$,
\item and for each $t\in [0,T)$, for all $x \in \Omega$ and $v\in \R^d$
\end{enumerate}
\begin{eqnarray}
&&f(t,x,v) = f_0(x-vt,v)\emph{\mbox{exp}}\left[-\int_{0}^t L[f(s,X_{s,t}(x,v),\cdot)](v)\:ds\right] \label{mildCO0}
\\ &&\quad\quad + \int_{0}^t \emph{\mbox{exp}}\left(-\int_s^t L[f(s',X_{s',s}(x,v),\cdot)](v)\:ds'\right)\nonumber
\\&&\quad\quad\quad\quad\quad\quad\quad \times Q^+[f(s,X_{s,t}(x,v),\cdot), f(s,X_{s,t}(x,v),\cdot)](v)\: ds.\nonumber
\end{eqnarray}
if $t \leq t_\partial(x,v)$ or else
\begin{eqnarray}
&&f(t,x,v) = f_\partial(t_\partial(x,v),x_\partial(x,v),v)\emph{\mbox{exp}}\left[-\int_{t_\partial(x,v)}^t L[f(s,X_{s,t}(x,v),\cdot)](v)\:ds\right] \label{mildCOtpartial}
\\ && \quad\quad+ \int_{t_\partial(x,v)}^t \emph{\mbox{exp}}\left(-\int_s^t L[f(s',X_{s',s}(x,v),\cdot)](v)\:ds'\right) \nonumber
\\&&\quad\quad\quad\quad\quad\quad\quad \times Q^+[f(s,X_{s,t}(x,v),\cdot), f(s,X_{s,t}(x,v),\cdot)](v)\: ds.\nonumber
\end{eqnarray}
\end{defi}
\bigskip

Now we state our result.

\bigskip
\begin{theorem}\label{theo:boundcutoff}
Let $\Omega$ be a $C^2$ open bounded domain in $\R^d$ with nowhere null normal vector and let $f_0$ be a non-negative continuous function on $\bar{\Omega} \times \R^d$. 
Let $B=\Phi b$ be a collision kernel satisfying $\eqref{assumptionB}$, with $\Phi$ satisfying $\eqref{assumptionPhi}$ or $\eqref{assumptionPhimol}$ and $b$ satisfying $\eqref{assumptionb}$ with $\nu < 0$. Let $f(t,x,v)$ be a mild solution of the Boltzmann equation in $\bar{\Omega}\times \R^d$ on some time intervalle $[0,T)$, $T \in (0,+\infty]$, which satisfies
\begin{itemize}
\item $f$ is continuous on $[0,T) \times \left(\bar{\Omega} \times \R^d-\Lambda_0\right)$ ($\Lambda_0$ grazing set defined by $\eqref{grazingset}$), $f(0,x,v) = f_0(x,v)$ and $M>0$ in $\eqref{massconservation}$;
\item if $\Phi$ satisfies $\eqref{assumptionPhi}$ with $\gamma \geq 0$ or if $\Phi$ satisfies $\eqref{assumptionPhimol}$, then $f$ satisfies $\eqref{assumption1}$;
\item if $\Phi$ satisfies $\eqref{assumptionPhi}$ with $\gamma < 0$, then $f$ satisfies $\eqref{assumption1}$ and $\eqref{assumption2}$.
\end{itemize}
Then for all $\tau \in (0,T)$ there exists $\rho >0$ and $\theta > 0$, depending on $\tau$, $E_f$ (and $L^{p_\gamma}_f$ if  $\Phi$ satisfies $\eqref{assumptionPhi}$ with $\gamma < 0$), such that for all $t \in [\tau,T)$ the solution $f$ is bounded from below, almost everywhere, by a global Maxwellian distribution with density $\rho$ and temperature $\theta$, i.e.
$$\forall t \in [\tau,T),\:\forall (x,v) \in \bar{\Omega}\times \R^d, \quad f(t,x,v) \geq \frac{\rho}{(2\pi \theta )^{d/2}}e^{-\frac{\abs{v}^2}{2\theta }}. $$
\end{theorem}
\bigskip

If we add the assumptions of uniform boundedness of $f_0$ and of the local mass and entropy of the solution $f$ we can use the arguments originated in \cite{PW} to construct explicitely the initial ``upheaval point'', without any compactness argument. We refer the reader to Section $\ref{subsec:constructiveupheavalpoint}$ which gives the following corollary.

\bigskip
\begin{cor}\label{cor:constructivecutoff}
Suppose that conditions of Theorem $\ref{theo:boundcutoff}$ are satisfied and further assume that $f_0$ is uniformly bounded from below
$$\forall (x,v) \in \Omega\times\R^d, \quad f_0(x,v) \geq \varphi(v) > 0,$$
and that $f$ has a bounded  local mass and entropy
\begin{eqnarray*}
R_f &=& \inf\limits_{(t,x)\in [0,T)\times \Omega}\int_{\R^d}f(t,x,v)\:dv >0
\\H_f &=& \sup\limits_{(t,x)\in [0,T)\times \Omega}\abs{\int_{\R^d}f(t,x,v)\emph{\mbox{log}}f(t,x,v)\:dv} <+\infty.
\end{eqnarray*}
Then conclusion of Theorem $\ref{theo:boundcutoff}$ holds true with the constants $\rho$ and $\theta$ being explicitely constructed in terms of $\tau$, $E_f$, $H_f$ and $L^{p_\gamma}_f$.
\end{cor}
\bigskip


\subsection{Exponential lower bound for non-cutoff collision kernels}\label{subsec:mainnoncutoff}

In the case of non-cutoff collision kernels ($0 \leq \nu <2$ in $\eqref{assumptionb}$), Grad's splitting does not make sense anymore and so we have to find a new way to define mild solutions to the Boltzmann equation $\eqref{BE}$. The splitting we are going to use is a standard one and it reads
\begin{eqnarray*}
Q(g,h) &=&  \int_{\R^d\times \mathbb{S}^{d-1}}\Phi\left(|v - v_*|\right)b\left( \mbox{cos}\theta\right)\left[h'g'_* - hg_*\right]dv_*d\sigma
\\     &=& Q^1_b(g,h) - Q^2_b(g,h),
\end{eqnarray*}
where we used the following definitions
\begin{eqnarray}
Q^1_b(g,h) &=& \int_{\R^d\times \mathbb{S}^{d-1}}\Phi\left(|v - v_*|\right)b\left( \mbox{cos}\theta\right)g'_*\left(h'-h\right)\: dv_*d\sigma,\nonumber
\\ Q^2_b(g,h) &=&  -\left(\int_{\R^d\times \mathbb{S}^{d-1}}\Phi\left(|v - v_*|\right)b\left( \mbox{cos}\theta\right)\left[g'_*-g_*\right]\:dv_*d\sigma\right)h \label{splittingQ1Q2}
\\ &=& S[g](v)h.\nonumber
\end{eqnarray}

We would like to use the properties we derived in the study of collision kernels with cutoff. Therefore we will consider additional splitting of $Q$.
\par For $\eps$ in $(0,\pi/4)$ we define a cutoff angular collision kernel
$$b^{CO}_\eps\left( \mbox{cos}\theta\right) = b\left( \mbox{cos}\theta\right)\mathbf{1}_{\abs{\theta} \geq \eps}$$
and a non-cutoff one
$$b^{NCO}_\eps\left( \mbox{cos}\theta\right) = b\left( \mbox{cos}\theta\right)\mathbf{1}_{\abs{\theta} \leq \eps}.$$

Considering the two collision kernels $B^{CO}_\eps = \Phi b^{CO}_\eps$ and $B^{NCO}_\eps = \Phi b^{NCO}_\eps$, we can combine Grad's splitting $\eqref{gradsplitting}$ applied to  $B^{CO}_\eps$ with the non-cutoff splitting $\eqref{splittingQ1Q2}$ applied to  $B^{NCO}_\eps$. This yields the splitting we shall use to deal with non-cutoff collision kernels,
\begin{equation}\label{noncutoffsplitting}
Q = Q^+_\eps - Q^-_\eps + Q^1_\eps - Q^2_\eps,
\end{equation}
where we use the shortened notations $Q^{\pm}_\eps = Q^{\pm}_{b^{CO}_\eps}$ and $Q^{i}_\eps = Q^{i}_{b^{NCO}_\eps}$, for $i=1,2$.

\bigskip
Thanks to the splitting $\eqref{noncutoffsplitting}$, we are able to define mild solutions to the Boltzmann equation with non-cutoff collision kernels. This is obtained by considering the Duhamel formula associated to the splitting $\eqref{noncutoffsplitting}$ along the characteristics (as in the cutoff case).

\bigskip
\begin{defi}\label{def:mildnoncutoff}
Let $f_0$ be a measurable function, non-negative almost everywhere on $\bar{\Omega}\times \R^d$.
\\ A measurable function $f = f(t, x, v)$ on $[0, T)\times\bar{\Omega}\times \R^d$ is a mild solution of the Boltzmann equation with non-cutoff angular collision kernel associated to the initial datum $f_0(x, v)$ if there exists $0 < \eps_0 <\pi/4$ such that for all $0<\eps<\eps_0$:
\begin{enumerate}
\item $f$ is non-negative on $\bar{\Omega}\times \R^d$,
\item for every $(t, x, v)$ in $\R^+\times \Omega\times \R^d$:
$$s \longmapsto L_\eps[f(t,X_{s,t}(x,v),\cdot)](v), \:\: s \longmapsto Q^+_\eps[f(t,X_{s,t}(x,v),\cdot), f(t,X_{s,t}(x,v),\cdot)](v)$$
$$s \longmapsto S_\eps[f(t,X_{s,t}(x,v),\cdot)](v), \:\: s \longmapsto Q^1_\eps[f(t,X_{s,t}(x,v),\cdot), f(t,X_{s,t}(x,v),\cdot)](v)$$
are in $L^1_{loc}([0,T))$,
\item and for each $t\in [0,T)$, for all $x \in \Omega$ and $v\in \R^d$
\end{enumerate}
\begin{eqnarray}
&&f(t,x,v) = f_0(x-vt,v)\emph{\mbox{exp}}\left[-\int_{0}^t \left(L_\eps + S_\eps\right)[f(s,X_{s,t}(x,v),\cdot)](v)\:ds\right] \label{mildNCO0}
\\ &&\quad\quad + \int_{0}^t \emph{\mbox{exp}}\left(-\int_s^t \left(L_\eps + S_\eps\right)[f(s',X_{s',s}(x,v),\cdot)](v)\:ds'\right)\nonumber
\\&&\quad\quad\quad\quad\quad\quad\quad \times \left(Q_\eps^+ + Q^1_\eps\right)[f(s,X_{s,t}(x,v),\cdot), f(s,X_{s,t}(x,v),\cdot)](v)\: ds.\nonumber
\end{eqnarray}
if $t \leq t_\partial(x,v)$ or else
\begin{eqnarray}
&&f(t,x,v) = f_\partial(t_\partial,x_\partial,v)\emph{\mbox{exp}}\left[-\int_{t_\partial(x,v)}^t \left(L_\eps+S_\eps\right)[f(s,X_{s,t}(x,v),\cdot)](v)\:ds\right] \label{mildNCOtpartial}
\\ && \quad\quad+ \int_{t_\partial(x,v)}^t \emph{\mbox{exp}}\left(-\int_s^t \left(L_\eps + S_\eps\right)[f(s',X_{s',s}(x,v),\cdot)](v)\:ds'\right) \nonumber
\\&&\quad\quad\quad\quad\quad\quad\quad \times \left(Q_\eps^+ + Q^1_\eps\right)[f(s,X_{s,t}(x,v),\cdot), f(s,X_{s,t}(x,v),\cdot)](v)\: ds,\nonumber
\end{eqnarray}
where $t_\partial=t_\partial(x,v)$ and $x_\partial=x_\partial(x,v)$ are defined by $\eqref{tpartial}$ and $\eqref{xpartial}$ respectively.
\end{defi}

\bigskip
Now we state our result.

\bigskip
\begin{theorem}\label{theo:boundnoncutoff}
Let $\Omega$ be a $C^2$ open bounded domain in $\R^d$ with nowhere null normal vector and let $f_0$ be a non-negative continuous function on $\bar{\Omega} \times \R^d$. 
Let $B=\Phi b$ be a collision kernel satisfying $\eqref{assumptionB}$, with $\Phi$ satisfying $\eqref{assumptionPhi}$ or $\eqref{assumptionPhimol}$ and $b$ satisfying $\eqref{assumptionb}$ with $\nu$ in $[0,2)$. Let $f(t,x,v)$ be a mild solution of the Boltzmann equation in $\bar{\Omega}\times \R^d$ on some time intervalle $[0,T)$, $T \in (0,+\infty]$, which satisfies
\begin{itemize}
\item $f$ is continuous on $[0,T) \times \left(\bar{\Omega} \times \R^d-\Lambda_0\right)$ ($\Lambda_0$ grazing set defined by $\eqref{grazingset}$), $f(0,x,v) = f_0(x,v)$ and $M>0$ in $\eqref{massconservation}$;
\item if $\Phi$ satisfies $\eqref{assumptionPhi}$ with $\gamma \geq 0$ or if $\Phi$ satisfies $\eqref{assumptionPhimol}$, then $f$ satisfies $\eqref{assumption1}$ and $\eqref{assumption3}$;
\item if $\Phi$ satisfies $\eqref{assumptionPhi}$ with $\gamma < 0$, then $f$ satisfies $\eqref{assumption1}$, $\eqref{assumption2}$ and $\eqref{assumption3}$.
\end{itemize}
Then for all $\tau \in (0,T)$ and for any exponent $K$ such that
$$K > 2\frac{\emph{\mbox{log}}\left(2+\frac{2\nu}{2-\nu}\right)}{\emph{\mbox{log}2}},$$

there exists $C_1, C_2 >0$, depending on $\tau$, $K$, $E_f$, $E'_f$, $W_f$ (and $L^{p_\gamma}_f$ if  $\Phi$ satisfies $\eqref{assumptionPhi}$ with $\gamma < 0$), such that
$$\forall t \in [\tau,T),\:\forall (x,v) \in \bar{\Omega}\times \R^d, \quad f(t,x,v) \geq C_1e^{-C_2\abs{v}^K}. $$
Moreover, in the case $\nu=0$, one can take $K=2$ (Maxwellian lower bound).
\end{theorem}
\bigskip

As in the angular cutoff case, if we further assume that $f_0$ presents no vacuum area and that $f$ has uniformly bounded local mass and entropy, our results are entirely constructive.

\bigskip
\begin{cor}\label{cor:constructivenoncutoff}
As for Corollary $\ref{cor:constructivecutoff}$, if $f_0$ is bounded uniformly from below as well as the local mass of $f$, the local entropy of $f$ is uniformly bounded from above then the conclusion of Theorem $\ref{theo:boundnoncutoff}$ holds true with constants being explicitely constructed in terms of $\tau$, $K$, $E_f$, $E'_f$, $W_f$, $H_f$, $L^{p_\gamma}_f$.
\end{cor}
\bigskip

\section{The cutoff case: a Maxwellian lower bound}\label{sec:cutoffcase}

In this section we are going to prove a Maxwellian lower bound for a solution to the Boltzmann equation $\eqref{BE}$ in the case where the collision kernel satisfies a cutoff property.
\par The strategy to tackle this result follows the main idea used in \cite{PW}\cite{Mo2}\cite{Bri2} which relies on finding an ``upheaval point'' (a first minoration uniform in time and space but localised in velocity) and spreading this bound, thanks to a spreading property of the $Q^+$ operator, in order to include larger and larger velocities.
\par As described in the introduction (Section $\ref{subsec:strategy}$), we need a method that translates the usual spreading argument in the case of our problem and combine it with a strict positivity of the diffusion process. Roughly speaking, either the characteristic we are looking at comes from the diffusion of the boundary or the spreading will be generated on a straight line as in \cite{Mo2}.

\bigskip 
Thus our study will be split into three parts, which are the next three subsections
\par The first step (Section $\ref{subsec:1stlocalisedbounds}$) is to partition the position and velocity spaces so that we have an immediate appearance of an ``upheaval point" in each of those partitions.
\par As discussed in the introduction, the standard spreading method fails in the case of characteristics trajectories bouncing against $\partial\Omega$. We therefore study of the effects of the Maxwellian diffusion (Section $\ref{subsec:diffusioneffects}$).
\par The last one (Section $\ref{subsec:maxwellboundcutoff}$) is to obtain uniform lower bounds less and less localised (and comparable to an exponential bound in the limit) in velocity. The strategy we use relies on the spreading property of the gain operator that we already used and follows the method in \cite{Mo2} combined with the treatment of localised bounds in \cite{Bri2}, together with the previous focus on the Maxwellian diffusion process.
\par A separated part, Section $\ref{subsec:constructiveupheavalpoint}$, is dedicated to a constructive approach to the Maxwellian lower bound.
\bigskip


\subsection{Initial localised lower bounds}\label{subsec:1stlocalisedbounds}

In Section $\ref{subsubsec:1stupheaval}$ we use the continuity of $f$ together with the conservation of total mass $\eqref{massconservation}$ and the uniform boundedness of the local energy to obtain a point in the phase space where $f$ is strictly positive. Then, thanks to the continuity of $f$, its Duhamel representation $\eqref{mildCOtpartial}-\eqref{mildCO0}$ and the spreading property of the $Q^+$ operator (Lemma $\ref{lem:Q+spread}$) we extend this positivity to high velocities at that particular point (Lemma $\ref{lem:positivity1}$).
\par Finally, in Section $\ref{subsubsec:1stlocalisedbounds}$, the free transport part of the solution $f$ will imply the immediate appearance of the localised lower bounds (Proposition $\ref{prop:upheaval}$).
\par Moreover we define constants that we will use in the next two subsections in order to have a uniform lower bound. 
\bigskip


\subsubsection{Controls on the gain and the loss operators}\label{subsubsec:controlgainloss}

We first introduce two lemmas, proven in \cite{Mo2}, that control the gain and loss terms in the Grad's splitting $\eqref{gradsplitting}$ we are using under the cutoff assumption.  One has an $L^\infty$ bound on the loss term (Corollary $2.2$ in \cite{Mo2}).

\bigskip
\begin{lemma}\label{lem:L}
Let $g$ be a measurable function on $\R^d$. Then
$$\forall v \in \R^d,\quad \abs{L[g](v)} \leq C_g^L\langle v \rangle^{\gamma^+},$$
where $C_g^L$ is defined by:
\begin{enumerate}
\item If $\Phi$ satisfies $\eqref{assumptionPhi}$ with $\gamma \geq 0$ or if $\Phi$ satisfies $\eqref{assumptionPhimol}$, then
$$C^L_g = \emph{\mbox{cst}}\: n_b C_\Phi e_g.$$
\item If $\Phi$ satisfies $\eqref{assumptionPhi}$ with $\gamma \in (-d,0)$, then
$$C^L_g = \emph{\mbox{cst}}\: n_b C_\Phi \left[e_g+ l^p_g\right],\quad p > d/(d+\gamma).$$
\end{enumerate}
\end{lemma}
\bigskip

The spreading property of $Q^+$ is given by the following lemma (Lemma $2.4$ in \cite{Mo2}), where we define
\begin{equation}\label{lb}
l_b = \inf_{\pi/4 \leq \theta \leq 3\pi/4}b\left(\mbox{cos}\:\theta\right).
\end{equation}

\bigskip
\begin{lemma}\label{lem:Q+spread}
Let $B=\Phi b$ be a collision kernel satisfying $\eqref{assumptionB}$, with $\Phi$ satisfying $\eqref{assumptionPhi}$ or $\eqref{assumptionPhimol}$ and $b$ satisfying $\eqref{assumptionb}$ with $\nu \leq 0$. Then for any $\bar{v} \in \R^d$, $0<r \leq R$, $\xi \in (0,1)$, we have

$$Q^+(\mathbf{1}_{B(\bar{v},R)},\mathbf{1}_{B(\bar{v},r)}) \geq \emph{\mbox{cst}}\: l_b c_\Phi r^{d-3} R^{3+\gamma} \xi^{\frac{d}{2}-1}\mathbf{1}_{B\left(\bar{v},\sqrt{r^2 + R^2}(1-\xi)\right)}.$$

As a consequence in the particular quadratic case $\delta = r = R$, we obtain

$$Q^+(\mathbf{1}_{B(\bar{v},\delta)},\mathbf{1}_{B(\bar{v},\delta)}) \geq \emph{\mbox{cst}}\: l_b c_\Phi \delta^{d+\gamma} \xi^{\frac{d}{2}-1}\mathbf{1}_{B\left(\bar{v},\delta\sqrt{2}(1-\xi)\right)},$$
for any $\bar{v} \in \R^d$ and $\xi \in (0,1)$.
\end{lemma}
\bigskip


\subsubsection{First ``upheaval'' point}\label{subsubsec:1stupheaval}

\bigskip
We start by the strict positivity of our function at one point for all velocities.

\bigskip
\begin{lemma}\label{lem:positivity1}
Let $f$ be the mild solution of the Boltzmann equation described in Theorem $\ref{theo:boundcutoff}$.
\\ Then there exist $\Delta >0$, $(x_1,v_1)$ in $\Omega \times \R^d$ such that for all $n \in \N$ there exist $r_n >0$, depending only on $n$, and $t_n(t),\:\alpha_n(t)>0$ such that
$$\forall t \in [0,\Delta],\:\forall x \in B\left(x_1,\frac{\Delta}{2^n}\right), \:\forall v \in \R^d, \quad f(t,x,v) \geq \alpha_n(t)\mathbf{1}_{B(v_1,r_n)}(v),$$
with $\alpha_0 >0$ independent of $t$ and the induction formula
$$\alpha_{n+1}(t) = C_Q\frac{r_n^{d+\gamma}}{4^{d/2-1}}\int_{t_n(t)}^{t}e^{-s C_L  \langle 2 r_n+\norm{v_1}\rangle^{\gamma^+}}\alpha^2_n(s)\:ds$$
where $C_Q=cst\: l_b c_\Phi$ is defined in Lemma $\ref{lem:Q+spread}$ and $C_L = cst\: n_b C_\Phi E_f$ (or $C_L = cst\: n_b C_\Phi (E_f+L^p_f)$) is defined in Lemma $\ref{lem:L}$, and
\begin{eqnarray}
&&r_0 = \Delta, \quad r_{n+1} = \frac{3\sqrt{2}}{4}r_n,\label{inductionrn1}
\\&&t_n(t)= \max\left\{0, t-\frac{\Delta}{2^{n+1}\left(\norm{v_1}+r_n\right)}\right\}.\label{definitiontn}
\end{eqnarray}
\end{lemma}
\bigskip

\begin{remark}
It is essentially the same method used to generate ``upheaval points" in \cite{Bri2}. The main difference being that we need to control the characteristics before they bounce against the boundary, whence the existence of the bound $t_n(t)$.
\end{remark}

\bigskip
\begin{proof}[Proof of Lemma $\ref{lem:positivity1}$]
The proof is an induction on $n$ and mainly follows the method in \cite{Bri2} Lemma $3.3$.

\bigskip
\textbf{Step $1$: Initialization}. We recall the assumptions that are made on the solution $f$ ($\eqref{massconservation}$ and assumption $\eqref{assumption1}$):
$$\forall t \in \R^+, \quad \int_{\Omega}\int_{\R^d} f(t,x,v)\:dxdv = M, \quad  \sup\limits_{(t,x)\in [0,T)\times \Omega}\int_{\R^d} \abs{v}^2f(t,x,v)\:dxdv = E_f,$$
with $M>0$ and $E_f<\infty$.
\\ Since $\Omega$ is bounded, and so is included in, say, $B(0,R_X)$, we also have that
$$\forall t \in \R^+,\quad \int_{\Omega}\int_{\R^d} \left(\abs{x}^2+ \abs{v}^2\right) f(t,x,v)\:dxdv \leq  \alpha = M R_X^2 + R_X E_f < +\infty.$$
Therefore, exactly the same computations as in \cite{Bri2} (step $1$ of the proof of Lemma $3.3$)  are applicable and lead to the existence of $(x_1,v_1)$ such that $f(0,x_1,v_1)>0$ and, by uniform continuity of $f$, to Lemma $\ref{lem:positivity1}$ in the case $n=0$.

\bigskip
\textbf{Step $2$: Proof of the induction}. We assume the conjecture is valid for $n$.
\\ Let $x$ be in $B(x_1,\Delta/2^{n+1})$, $v$ in $B(0,\norm{v_1} + 2r_n)$ and $t$ in $[0,\Delta]$.

\bigskip
We have straightforwardly that
$$\forall s \in \left[t_n(t),t\right], \quad \norm{x_1- (x - (t-s)v)} \leq \frac{\Delta}{2^{n}}.$$
Besides, we have that $B(x_1,\Delta) \subset \Omega$ and therefore the characteristic line $\left(X_{s,t}(x,v)\right)_{s\in[t_n(t),t]}$ stays in $\Omega$. This implies that $t_\partial(x,v)>t_n(t)$.

\bigskip
We thus use the fact that $f$ is a mild solution to write $f(t,x,v)$ under its Duhamel form starting at $t_n(t)$ without contact with the boundary $\eqref{mildCO0}$. The control we have on the $L$ operator, Lemma $\ref{lem:L}$, allows us to bound from above the second integral term (the first term is positive).  Moreover, this bound on $L$ is independent on $t$, $x$ and $v$ since it only depends on an upper bound on the energy $e_{f(t,x,\cdot)}$ (and its local $L^p$ norm $l^p_{f(t,x,\cdot)}$) which is uniformly bounded by $E_f$ (and by $L^p_f$). This yields 

\begin{equation}\label{inductionpositivity}
f(t,x,v) \geq  \int_{t_n(t)}^{t}e^{-s C_L \langle \norm{v_1}+ 2r_n \rangle^{\gamma^+}}  Q^+ \left[f(s,X_{s,t}(x,v),\cdot),f(s,X_{s,t}(x,v),\cdot)\right]\left(v\right)\:ds,
\end{equation}
where $C_L = cst\: n_b C_\Phi E_f$ (or $C_L = cst\: n_b C_\Phi (E_f+L^p_f)$), see Lemma $\ref{lem:L}$, and we used $\norm{v} \leq 2r_n + \norm{v_1}$.

\bigskip
We already saw that $\left(X_{s,t}(x,v)\right)_{s\in[t_n(t),t]}$ stays in $B(x_1,\Delta/2^n)$. Therefore, by calling $v_*$ the integration parametre in the operator $Q^+$ we can apply the induction property to $f(s,X_{s,t}(x,v),v_*)$ which implies, in $\eqref{inductionpositivity}$,

$$f(t,x,v) \geq  \int_{t_n(t)}^{t}e^{-s C_L \langle \norm{v_1}+2r_n \rangle^{\gamma^+}}\alpha_n^2(s)  Q^+\left[\mathbf{1}_{B(v_1,r_n)},\mathbf{1}_{B(v_1,r_n)}\right]\:ds (v).$$

\bigskip
Applying the spreading property of $Q^+$, Lemma $\ref{lem:Q+spread}$, with $\xi = 1/4$ gives us the expected result for the step $n+1$ since $B(v_1,r_{n+1}) \subset B(0,\norm{v_1}+2r_n)$.
\end{proof}
\bigskip


\subsubsection{Partition of the phase space and first localised lower bounds}\label{subsubsec:1stlocalisedbounds}

We are now able to prove the immediate appearance of localised ``upheaval points". We emphasize here that the following proposition is proven with exactly the same arguments as in \cite{Bri2}.

\bigskip
\begin{prop}\label{prop:upheaval}
Let $f$ be the mild solution of the Boltzmann equation described in Theorem $\ref{theo:boundcutoff}$ and consider $x_1$, $v_1$ constructed in Lemma $\ref{lem:positivity1}$.
\\Then there exists $\Delta >0$ such that for all $0<\tau_0 \leq \Delta$, there exists $\delta_T(\tau_0)$, $\delta_X(\tau_0)$, $\delta_V(\tau_0)$,  $R_{min}(\tau_0)$, $a_0(\tau_0) > 0$ such that for all $N$ in $\N$ there exists $N_X$ in $\N^*$ and $x_{1},\dots,x_{N_{X}}$ in $\Omega$  and $v_1,\dots,v_{N_X}$ in $B(0,R_{min}(\tau_0))$ and
\begin{itemize}
\item $\bar{\Omega} \subset \bigcup\limits_{1 \leq i \leq N_X}B\left(x_i,\delta_X(\tau_0)/2^{N}\right)$;
\item $\forall t \in [\tau_0,\delta_T(\tau_0)], \:\forall x \in B(x_i,\delta_X(\tau_0)),\forall v \in \R^d,$ 
$$f(t,x,v) \geq a_0(\tau_0)\mathbf{1}_{B\left(v_i,\delta_V(\tau_0,N)\right)}(v).$$
\end{itemize}
\end{prop}
\bigskip

\begin{proof}[Proof of Proposition $\ref{prop:upheaval}$]
We are going to use the free transport part of the Duhamel form of $f$ $\eqref{mildCO0}$, to create localised lower bounds out of the one around $(x_1,v_1)$ in Lemma $\ref{lem:positivity1}$.

\bigskip
$\Omega$ is bounded so let us denote its diameter by $d_{\Omega}$.
\par Take $\tau_0$ in $(0,\Delta]$. Let $n$ be large enough such that $r_n \geq 2 d_{\Omega}/\tau_0+\norm{v_1}$, where $r_n$ is defined by $\eqref{inductionrn1}$ in Lemma $\ref{lem:positivity1}$. It is possible since $(r_n)$ increases to infinity. Now, define $R_{min}(\tau_0) = 2 d_{\Omega}/\tau_0$.
\par Thanks to Lemma $\ref{lem:positivity1}$ applied to this particular $n$ we have that
\begin{equation}\label{initialbound}
\forall t \in \left[\frac{\tau_0}{2},\Delta\right],\:\forall x \in B(x_1,\Delta/2^n), \quad f(t,x,v) \geq \alpha_n\left(\frac{\tau_0}{2}\right)\mathbf{1}_{B(v_1,r_n)}(v),
\end{equation}
where we used the fact that $\alpha_n(t)$ is an increasing function.

\bigskip
Define
$$a_0(\tau_0) = \frac{1}{2}\alpha_n\left(\frac{\tau_0}{2}\right) e^{-\frac{\tau_0}{2} C_L \langle \frac{2 d_{\Omega}}{\tau_0} \rangle^{\gamma^+}}.$$

\bigskip
\textbf{Definition of the constants.}
We notice that for all $x$ in $\partial\Omega$ we have that $n(x)\cdot (x-x_1) >0$, because $\Omega$ has nowhere null normal vector by hypothesis. But the function 
$$x \longmapsto n(x)\cdot \frac{x-x_1}{\norm{x-x_1}}$$
 is continuous (since $\Omega$ is $C^2$) on the compact $\partial{\Omega}$ and therefore has a minimum that is atteined at a certain $X(x_1)$ on $\partial\Omega$.
\par Hence, 
\begin{equation}\label{lambdax1}
\forall x \in \partial\Omega, \quad n(x)\cdot \frac{x-x_1}{\norm{x-x_1}} \geq n(X(x_1))\cdot \frac{X(x_1)-x_1}{\norm{X(x_1)-x_1}}= 2 \lambda(x_1) >0.\end{equation}

\bigskip
To shorten following notations, we define on $\bar{\Omega} \times \left(\R^d-\{0\}\right)$ the function
\begin{equation}\label{Phi}
\Phi(x,v) = n\left(x+t\left(x,\frac{v}{\norm{v}}\right)\frac{v}{\norm{v}}\right),
\end{equation}
where we defined $t(x,v) = \min\{t\geq 0: x+tv \in \partial\Omega\}$, the first time of contact against the boundary of the forward characteristic $(x+sv)_{s\geq 0}$ defined for $v\neq 0$ and continuous on $\bar{\Omega}\times\left(\R^d-\{0\}\right)$ (see \cite{Bri2} Lemma $5.2$ for instance).

\bigskip
We denote $d_1$ to be half of the distance from $x_1$ to $\partial\Omega$. We define two sets included in $[0,\Delta]\times\bar{\Omega}\times\R^d$:
$$\Lambda^{(1)} = [0,\Delta]\times B(x_1,d_1) \times \R^d$$
and
$$\Lambda^{(2)} = \left\{(t,x,v)\notin \Lambda^{(1)},\quad \norm{v}\geq \frac{d_1}{\tau_0} \quad\mbox{and} \quad \Phi(x,v)\cdot \frac{v}{\norm{v}} \geq \lambda(x_1)\right\}$$

\bigskip
By continuity of $t(x,v)$ and of $n$ (on $\partial\Omega$), we have that
$$\Lambda = \Lambda^{(1)} \cap \Lambda^{(2)}$$
is compact and does not intersect the grazing set $[0,\Delta]\times \Lambda_0$ defined by $\eqref{grazingset}$. Therefore, $f$ is continuous in $\Lambda$ and thus is uniformly continuous on $\Lambda$. Hence, there exist $\delta_T'(\tau_0)$, $\delta'_X(\tau_0)$, $\delta'_V(\tau_0) >0$ such that
$$\forall (t,x,v),\:(t',x',v')\in\Lambda,\quad |t-t'|\leq \delta'_T(\tau_0),\: \norm{x-x'}\leq \delta'_X(\tau_0),\: \norm{v-v'}\leq \delta'_V(\tau_0),$$
\begin{equation}\label{uniformcontinuity}
\abs{f(t,x,v)- f(t',x',v')} \leq a_0(\tau_0).
\end{equation}

\bigskip
The map $\Phi$ (defined by $\eqref{Phi}$) is uniformly continuous on the compact $[0,\Delta]\times\bar{\Omega}\times \mathbb{S}^{d-1}$ and therefore there exist $\delta_T''(\tau_0)$, $\delta''_X(\tau_0)$, $\delta''_V(\tau_0) >0$ such that
$$\forall (t,x,v),\:(t',x',v')\in \Lambda^{(2)},\quad |t-t'|\leq \delta''_T(\tau_0),\: \norm{x-x'}\leq \delta''_X(\tau_0),\: \norm{v-v'}\leq \delta''_V(\tau_0),$$
\begin{equation}\label{uniformlambda2}
\abs{ \Phi(x,v)-  \Phi(x',v')} \leq \frac{\lambda(x_1)}{2}.
\end{equation}

\bigskip
We conclude our definitions by taking 
\begin{eqnarray*}
\delta_T(\tau_0)&=& \min\left(\Delta,\: \tau_0 +\delta_T'(\tau_0),\: \tau_0+\delta_T''(\tau_0)\right),
\\ \delta_X(\tau_0)&=& \min\left(\frac{\Delta}{2^n},\:\delta_X'(\tau_0),\:\delta_X''(\tau_0),\:d_1/2\right),
\\ \delta_V(\tau_0) &=&\min\left(r_n,\:\delta_V'(\tau_0),\:\frac{d_1}{2\tau_0}\delta_V''(\tau_0),\:\frac{\lambda(x_1)}{2}\right).
\end{eqnarray*}

\bigskip
\textbf{Proof of the lower bounds.}
We take $N \in \N$ and notice that $\bar{\Omega}$ is compact and therefore there exists $x_{1},\dots,x_{N_{X}}$ in $\Omega$ such that $\bar{\Omega} \subset \bigcup\limits_{1 \leq i \leq N_{X}}B\left(x_i,\delta_X(\tau_0)/2^{N}\right)$. Moreover, we construct them such that $x_{1}$ is the one defined in Lemma $\ref{lem:positivity1}$ and we then take $v_1$ to be the one defined in Lemma $\ref{lem:positivity1}$. We define 
$$\forall i \in \{2,\dots,N_X\}, \quad v_i =\frac{2}{\tau_0}(x_i - x_1).$$
Because $\Omega$ is convex we have that 
\begin{eqnarray*}
X_{\tau_0/2,\tau_0}(x_i,v_i) &=& x_1,
\\ V_{\tau_0/2,\tau_0}(x_i,v_i) &=& v_i.
\end{eqnarray*}

\bigskip
The latter equalities imply that there is no contact with $\partial\Omega$ between times $\tau_0/2$ and $\tau_0$ when starting from $x_1$ to go to $x_i$ with velocity $v_i$.
Using the fact that $f$ is a mild solution of the Boltzmann equation, we write it under its Duhamel form without contact  $\eqref{mildCO0}$, but starting at $\tau_0/2$. We drop the last term which is positive. As in the proof of Lemma $\ref{lem:positivity1}$ we can control the $L$ operator appearing in the first term in the right-hand side of $\eqref{mildCO0}$ (corresponding to the free transport).

\begin{eqnarray*}
f(\tau_0,x_i,v_i) &\geq& f\left(\frac{\tau_0}{2},x_1,v_i\right)e^{-\frac{\tau_0}{2}C_L \langle \frac{2}{\tau_0}(x_i-x_1)\rangle^{\gamma^+}}
\\ &\geq& \alpha_n\left(\frac{\tau_0}{2}\right)e^{-\frac{\tau_0}{2}C_L \langle \frac{2d_{\Omega}}{\tau_0}\rangle^{\gamma^+}}\mathbf{1}_{B(v_1,r_n)}(v_i)
\\ &\geq& 2 a_0(\tau_0)\mathbf{1}_{B(v_1,r_n)}(v_i),
\end{eqnarray*}
where we used $\eqref{initialbound}$ for the second inequality. We see here that $v_i$ belongs to $B(0,R_{min}(\tau_0))$ and that $B(0,R_{min}(\tau_0)) \subset B(v_1,r_n)$ and therefore

\begin{equation}\label{finpositivity}
f(\tau_0,x_i,v_i) \geq 2a_0(\tau_0).
\end{equation}

\bigskip
We first notice that $(\tau_0,x_i,v_i)$ belongs to $\Lambda$ since either $x_i$ belongs to $B(x_1,d_1)$ or $\norm{x_1-x_i} \geq d_1$ but by definition of $v_i$ and $\lambda(x_1)$ (see $\eqref{lambdax1}$),
$$n\left(x_i+t\left(x_i,\frac{v_i}{\norm{v_i}}\right)\frac{v_i}{\norm{v_i}}\right)\cdot \frac{v_i}{\norm{v_i}} \geq 2\lambda(x_1)$$
and
$$\norm{v_i} = \frac{2}{\tau_0}\norm{x_i-x_1} \geq \frac{2}{\tau_0}d_1.$$

\bigskip
We take $t$ in $[\tau_0,\delta_T(\tau_0)]$, $x$ in $B(x_i,\delta_X(\tau_0))$ and $v$ in $B(v_i,\delta_V(\tau_0))$ and we will prove that $(t,x,v)$ also belongs to $\Lambda$.
\par If $x_i$ belongs to $B(x_1,d_1/2)$ then since $\delta_X(\tau_0)\leq d_1/2$,
$$\norm{x-x_1} \leq \frac{d_1}{2} + \norm{x-x_i} \leq d_1$$
and $(t,x,v)$ thus belongs to $\Lambda^{(1)} \subset \Lambda$.
\par In the other case where $\norm{x_1-x_i}\geq d_1/2$ we first have that 
$$\norm{v_i} = \frac{2}{\tau_0}\norm{x_i-x_1} \geq \frac{d_1}{\tau_0}.$$
And also
$$\norm{\frac{v}{\norm{v}}-\frac{v_i}{\norm{v_i}}}\leq \frac{2}{\norm{v_i}}\norm{v-v_i}= \frac{\tau_0}{\norm{x_i-x_1}}\delta_V(\tau_0)\leq \frac{2\tau_0}{d_1}\delta_V(\tau_0)\leq \delta_V''(\tau_0).$$
The latter inequality combined with $\eqref{uniformlambda2}$ and that $\abs{t-\tau_0}\leq \delta_T''(\tau_0)$ and $\norm{x-x_i}\leq \delta''_X(\tau_0)$ yields
$$\abs{\Phi(x,v)- \Phi(x_i,v_i)} \leq \frac{\lambda(x_1)}{2},$$
which in turn implies
\begin{eqnarray*}
\Phi(x,v) \cdot \frac{v}{\norm{v}} &\geq& \Phi(x_i,v_i)\cdot\frac{v_i}{\norm{v_i}} + \Phi(n,v)\cdot(v-v_i) + \left(\Phi(x,v)-\Phi(x_i,v_i)\right)\cdot v_i
\\&\geq& 2\lambda(x_1)-\norm{v-v_i} - \abs{\Phi(x,v)-\Phi(x_i,v_i)}
\\&\geq& \lambda(x_1),
\end{eqnarray*}
so that $(t,x,v)$ belongs to $\Lambda^{(2)}$.

\bigskip
We can now conclude the proof.
\par We proved that $(\tau_0,x_i,v_i)$ belongs to $\Lambda$ and that for all  $t$ in $[\tau_0,\delta_T(\tau_0)]$, $x$ in $B(x_i,\delta_X(\tau_0))$ and $v$ in $B(v_i,\delta_V(\tau_0))$, $(t,x,v)$ belongs to $\Lambda$. By definition of the constants, $(t-\tau_0,x-x_i,v-v_i)$ satisfies the inequality of the uniform continuity of $f$ on $\Lambda$ $\eqref{uniformcontinuity}$. Combining this inequality with $\eqref{finpositivity}$, the lower bound at $(\tau_0,x_i,v_i)$, we have that 
$$f(t,x,v) \geq a_0(\tau_0).$$

\end{proof}
\bigskip

\begin{remark}\label{rem:tau0=0}
In order to lighten our presentation and because $\tau_0$ can be chosen as small as one wants, we will only study the case of solutions to the Boltzmann equation which satisfies Proposition $\ref{prop:upheaval}$ at $\tau_0=0$. Then we will immediatly create the exponential lower bound after at $\tau_1$ for all $\tau_1>0$. Then we apply the latter result to $F(t,x,v) = f(t+\tau_0,x,v)$ to obtain the exponential lower bound for $f$ at time $\tau_0+\tau_1$ which can be chosen as small as one wants.
\end{remark}
\bigskip


\subsection{Global strict positivity of the Maxwellian diffusion}\label{subsec:diffusioneffects}

In this subsection, we focus on the positivity of the Maxwellian diffusion boundary condition. More precisely, we prove that the boundary of the domain $\Omega$ diffuses in all directions a minimal strictly positive quantity.

\bigskip
\begin{prop}\label{prop:positivediffusion}
Let $f$ be the mild solution of the Boltzmann equation described in Theorem $\ref{theo:boundcutoff}$ and consider $\Delta>0$ constructed in Proposition $\ref{prop:upheaval}$.
\\ Then for all $\tau_0$ in $(0,\Delta]$, there exists $b_\partial(\tau_0) >0$ such that
$$\forall t \in [\tau_0,\Delta],\:\forall x_\partial \in \partial\Omega, \quad \left[\int_{v_*\cdot n(x_\partial)>0} f(t,x_\partial,v_*)\left(v_*\cdot n(x_\partial)\right)\:dv_*\right] > b_\partial(\tau_0).$$
\end{prop}

\bigskip
\begin{proof}[Proof of Proposition $\ref{prop:positivediffusion}$]
Let $\tau_0$ be in $(0,\Delta]$, take $t$ in $[\tau_0,\Delta]$ and $x_\partial$ on $\partial\Omega$.
\par We consider $(x_1,v_1)$ constructed in Lemma $\ref{lem:positivity1}$ and will use the same notations as in the proof of Proposition $\ref{prop:upheaval}$. 

\bigskip
In the spirit of the proof of Proposition $\ref{prop:upheaval}$ we define 
$$v_\partial = \frac{2}{t}(x_\partial-x_1),$$
which gives, because $\partial\Omega$ is convex, that $\left(X_{t/2,s}\right)_{t/2\leq s \leq t}$ is a straight line not intersecting $\partial\Omega$ apart at time $s=t$. We can thus write $f$ under its Duhamel form without contact $\eqref{mildCO0}$ starting at $t/2$. We keep only the first term on the right hand-side and control the operator $L$ by Lemma $\ref{lem:L}$.
\begin{eqnarray*}
f(t,x_\partial,v_\partial) &\geq& f\left(\frac{t}{2},x_1,v_\partial\right)e^{-\frac{t}{2}C_L \langle \frac{2}{t}(x_\partial-x_1)\rangle^{\gamma^+}}
\\ &\geq& \alpha_n\left(\frac{t}{2}\right)e^{-\frac{t}{2}C_L \langle \frac{2d_{\Omega}}{t}\rangle^{\gamma^+}}\mathbf{1}_{B(v_1,r_n)}(v_\partial)
\\ &\geq& \alpha_n\left(\frac{\tau_0}{2}\right)e^{-\frac{\Delta}{2}C_L \langle \frac{2d_{\Omega}}{\tau_0}\rangle^{\gamma^+}}\mathbf{1}_{B(v_1,r_n)}(v_\partial),
\end{eqnarray*}
where we used $\eqref{initialbound}$ for the second inequality since $t/2$ belongs to $[\tau_0/2,\Delta]$. Note that we choose $n$ exactly as in the proof of Proposition $\ref{prop:upheaval}$ and, for the same reasons, $v_\partial$ thus belongs to $B(v_1,r_n)$ which implies
$$f(t,x_\partial,v_\partial) \geq  \alpha_n\left(\frac{\tau_0}{2}\right)e^{-\frac{\Delta}{2}C_L \langle \frac{2d_{\Omega}}{\tau_0}\rangle^{\gamma^+}}.$$
Here again, the continuity of $f$ away from the grazing set implies the existence of $\delta(\tau_0)$ independent of $t$, $x$ and $v_\partial$ such that
\begin{equation}\label{positivitydiffusion1}
\forall v_* \in B(v_\partial,\delta(\tau_0)), \quad f(t,x_\partial,v_*) \geq  \frac{1}{2}\alpha_n\left(\frac{\tau_0}{2}\right)e^{-\frac{\Delta}{2}C_L \langle \frac{2d_{\Omega}}{\tau_0}\rangle^{\gamma^+}}= A(\tau_0)>0.
\end{equation}

\bigskip
We now deal with the scalar product $v_*\cdot n(x_\partial)$ appearing in the Maxwellian diffusion.
\par We notice that for all $x$ in $\partial\Omega$ we have that $n(x)\cdot (x-x_1) >0$, because $\Omega$ has nowhere null normal vector by hypothesis. But the function $x \longmapsto n(x)\cdot (x-x_1)$ is continuous (since $\Omega$ is $C^2$) on the compact $\partial{\Omega}$ and therefore has a minimum that is atteined at a certain $X(x_1)$ on $\partial\Omega$.
\par Hence, 
$$\forall x \in \partial\Omega, \quad (x-x_1)\cdot n(x) \geq n(X(x_1))\cdot (X(x_1)-x_1)= 2 B(x_1) >0.$$
We define $\delta'(x_1)= B(x_1)$ and a mere Cauchy-Schwarz inequality implies that
$$\forall x \in \partial\Omega,\:\forall v_* \in B\left((x-x_1),\delta'(x_1)\right), \quad v_*\cdot n(x) \geq B(x_1) >0,$$
which in turns implies
\begin{equation}\label{positivitydiffusion2}
\forall t \in [\tau_0,\Delta],\:\forall x \in \partial\Omega,\:\forall v_* \in B\left(\frac{2}{t}(x-x_1),\frac{2}{\Delta}\delta'(\tau_0)\right), \quad v\cdot n(x) \geq \frac{2}{\Delta}B(x_1).
\end{equation}

\bigskip
To conlude we combine $\eqref{positivitydiffusion1}$ and $\eqref{positivitydiffusion2}$ at point $x_\partial$ inside the integrale to get, with $\delta''(\tau_0) = \min\left(\delta(\tau_0),2\delta'(x_1)/\Delta\right)$,

\begin{eqnarray*}
\int_{v_*\cdot n(x_\partial)>0} f(t,x_\partial,v_*)\left(v_*\cdot n(x_\partial)\right)\:dv_* &\geq& \int_{v_*\in B(v_\partial,\delta''(\tau_0))} f(t,x_\partial,v_*)\left(v_*\cdot n(x_\partial)\right)\:dv_*
\\ &\geq& \frac{2}{\Delta}A(\tau_0)B(x_1)\abs{B(v_\partial,\delta''(\tau_0))}
\\&=& \frac{2}{\Delta} A(\tau_0)B(x_1)\abs{B(0,\delta''(\tau_0))}.
\end{eqnarray*}
This yields the expected result with 
\begin{equation}\label{finalboundiffusion}
b_\partial(\tau_0) =2 A(\tau_0)B(x_1)\abs{B(0,\delta''(\tau_0))}/\Delta.
\end{equation}
\end{proof}
\bigskip

\begin{remark}\label{rem:tau0=0diffusion}
For now on, $\tau_0$ used in Proposition $\ref{prop:positivediffusion}$ will be the same as the one used in Proposition $\ref{prop:upheaval}$. Moreover, as already mentioned and explained in Remark $\ref{rem:tau0=0}$, we will consider in the sequel that $\tau_0=0$.
\end{remark}
\bigskip


\subsection{Spreading of the initial localised bounds and global Maxwellian lower bound}\label{subsec:maxwellboundcutoff}

In Section $\ref{subsec:1stlocalisedbounds}$ we proved the immediate appearance of initial lower bounds that are localised in space and in velocity. The present subsection aims at increasing these lower bounds in a way that larger and larger velocities are taken into account, and compare it to an exponential bound in the limit.
\bigskip


\subsubsection{Spreading of the initial ``upheaval points''}\label{subsubsec:spreadingmethod}

First, we pick $N$ in $\N^*$, construct $\delta_V$, $R_{min}$ and cover $\bar{\Omega}$ with $\bigcup_{1 \leq i \leq N_X}B(x_i,\delta_X/2^{N})$ as in Proposition $\ref{prop:upheaval}$ where we dropped the dependencies in $\tau_0$ and $N$.
\\Then for any sequence $\left(\xi_n\right)$ in $(0,1)$ and for all $\tau>0$ we define three sequences in $\R^+$ by induction. First,

\begin{equation} \label{rn}
\left\{\begin{array}{rl} \displaystyle{r_0} &\displaystyle{= \delta_V} \vspace{2mm}\\\vspace{2mm} \displaystyle{r_{n+1} }&\displaystyle{= \sqrt{2}\left(1-\xi_n\right)r_n.} \end{array}\right.
\end{equation}
Second, with the notation $\tilde{r}_n = R_{min} + r_n$,
\begin{equation}\label{bn}
\left\{\begin{array}{rl} \displaystyle{b_0(\tau) }&\displaystyle{= b_\partial} \vspace{2mm}\\\vspace{2mm} \displaystyle{b_{n+1}(\tau) }&\displaystyle{= b_\partial e^{-C_L \tau \langle\tilde{r}_{n+1} \rangle^{\gamma^+}} \frac{e^{-\frac{\left(\tilde{r}_{n+1}\right)^2}{2T_{\partial}}}}{\left(2\pi\right)^{\frac{d-1}{2}}T_{\partial}^{\frac{d+1}{2}}},}\end{array}\right.
\end{equation}
where $b_\partial$ defined in Proposition $\ref{prop:positivediffusion}$ and $C_L$ in Lemma $\ref{lem:positivity1}$.
And finally, with $C_Q$ being defined in Lemma $\ref{lem:positivity1}$,
\begin{equation}\label{an}
\left\{\begin{array}{rl} \displaystyle{a_0(\tau) }&\displaystyle{= a_0} \vspace{2mm}\\\vspace{2mm} \displaystyle{a_{n+1}(\tau) }&\displaystyle{= \min\left(a_n(\tau),b_n(\tau)\right)^2C_Qr_n^{d+\gamma}\xi_{n+1}^{d/2-1}\frac{\tau}{2^{n+2}\tilde{r}_{n+1}}e^{-C_L \frac{\tau \langle \tilde{r}_{n+1} \rangle^{\gamma^+}}{2^{n+2}\tilde{r}_{n+1}}}.}\end{array}\right.
\end{equation}

\bigskip
We express the spreading of the lower bound in the following proposition.

\bigskip
\begin{prop}\label{prop:spread}
Let $f$ be the mild solution of the Boltzmann equation described in Theorem $\ref{theo:boundcutoff}$ and suppose that $f$ satisfies Proposition $\ref{prop:upheaval}$ with $\tau_0=0$.
\\Consider $0<\tau\leq \delta_T$ and $N$ in $\N$. Let $(x_i)_{i \in \{1,\dots,N_X\}}$ and $(v_i)_{i \in \{1,\dots,N_X\}}$ be given as in Proposition $\ref{prop:upheaval}$ with $\tau_0=0$.
\\Then for all $n$ in $\{0,\dots,N\}$ we have that the following holds:

$$\forall t \in \left[\tau-\frac{\tau}{2^{n+1}\tilde{r}_{n}},\tau\right],\:\forall x \in B\left(x_i,\frac{\delta_X}{2^n}\right), \quad f(t,x,v) \geq \min\left(a_n(\tau),b_n(\tau)\right) \mathbf{1}_{B(v_i,r_n)}(v),$$
where $(r_n)$, $(a_n)$ and $(b_n)$ are defined by $\eqref{rn}$-$\eqref{an}$-$\eqref{bn}$.
\end{prop}
\bigskip

\begin{proof}[Proof of Proposition $\ref{prop:spread}$]
 We are interested in immediate appearance of a lower bound so we can always choose $\delta_t \leq \delta_X$ and also $R_{min} \geq 1$.
 \par This Proposition will be proved by induction on $n$, the initialisation is simply Proposition $\ref{prop:upheaval}$ so we consider the case where the proposition is true for $n<N$.

\bigskip
Take $t\in [\tau-\tau/(2^{n+2}\tilde{r}_{n+1}),\tau]$, $x \in B\left(x_i,\delta_X/2^{n+1}\right)$ and $v\in B(v_i,2r_{n+1})$.
\par There are two possible cases depending on $t_\partial(x,v) \leq t$ or $t_\partial(x,v)> t$. We remind here that $t_\partial$ (see $\eqref{tpartial}$) is the first time of contact of the backward trajectory starting at $x$ with velocity $v$.

\bigskip
\textbf{$1^{st}$ case: no contact against the boundary from $0$ to $t$}. We can therefore use the Duhamel formula without contact $\eqref{mildCO0}$ and bound it by the second term on the right-hand side (every term being positive). Then Lemma $\ref{lem:L}$ on the linear operator $L$ and the fact that 
$$\norm{v} \leq \norm{v_i} + r_{n+1} \leq R_{min} +  r_{n+1}=\tilde{r}_{n+1},$$
gives the following inequality

\begin{equation}\label{inductionCO}
\begin{split}
f(t,x,v) \geq \int_{\tau-\frac{\tau}{2^{n+1}\tilde{r}_{n+1}}}^{\tau-\frac{\tau}{2^{n+2}\tilde{r}_{n+1}}} e^{-C_L (t-s)\langle \tilde{r}_{n+1} \rangle^{\gamma^+}} Q^+ \left[f(s,X_{s,t}(x,v),\cdot),f(s,X_{s,t}(x,v),\cdot)\right]\:ds.
\end{split}
\end{equation}
\bigskip

We now see that for all $s$ in the interval considered in the integral above
$$ \norm{x_i-X_{s,t}(x,v)}\leq \norm{x_i-x} + \abs{t-s}\norm{v} \leq \frac{\delta_X}{2^{n+1}} + \frac{\tau}{2^{n+1}}\leq \frac{\delta_X}{2^{n}},$$
where we used that $\delta_T \leq \delta_X$. We can therefore apply the induction hypothesis to $f(t,X_{s,t}(x,v),v_*)$ into $\eqref{inductionCO}$, where we denoted by $v_*$ the integration parameter in $Q^+$. This yields

\begin{equation*}
\begin{split}
&f(t,x,v) \geq  
\\&\quad a_n\left(\tau\right)^2 e^{-C_L \frac{\tau \langle \tilde{r}_{n+1} \rangle^{\gamma^+}}{2^{n+2}\tilde{r}_{n+1}}}\left(\int_{\tau-\frac{\tau}{2^{n+1}\tilde{r}_{n+1}}}^{\tau-\frac{\tau}{2^{n+2}\tilde{r}_{n+1}}}  Q^+ \left[\mathbf{1}_{B(v_i,r_n)},\mathbf{1}_{B(v_i,r_n)}\right]\:ds\right)(v).
\end{split}
\end{equation*}

\bigskip
Applying the spreading property of $Q^+$, Lemma $\ref{lem:Q+spread}$, with $\xi = \xi_{n+1}$ gives us the expected result for the step $n+1$ with the lower bound $a_{n+1}(\tau)$.

\bigskip
\textbf{$2^{nd}$ case: there is at least one contact before time $t$.} In that case we have that $t_\partial(x,v)\leq t$ and we can use the Duhamel formula with contact $\eqref{mildCOtpartial}$ for $f$. Both terms on the right-hand side are positive so we can lower bound $f(t,x,v)$ by the first one. Denoting $x_\partial=x_\partial(x,v)$ (see Definition $\eqref{xpartial}$), this yields
\begin{eqnarray}
&&f(t,x,v)\geq e^{-C_L t  \langle\tilde{r}_{n+1} \rangle^{\gamma^+}}f_\partial\left(t_\partial(x,v),x_\partial(x,v),v\right) \nonumber
\\&&\geq  \frac{e^{-\frac{\left(\tilde{r}_{n+1}\right)^2}{2T_{\partial}}}}{\left(2\pi\right)^{\frac{d-1}{2}}T_{\partial}^{\frac{d+1}{2}}}e^{-C_L \tau \langle\tilde{r}_{n+1} \rangle^{\gamma^+}}\left[\int_{v_*\cdot n(x_\partial)>0} f(t,x_\partial,v_*)\left(v_*\cdot n(x_\partial)\right)\:dv_*\right]. \label{2ndcase1}
\end{eqnarray}
where we used the definition of $f_\partial$ $\eqref{fpartial}$.

\bigskip
Thanks to the previous subsection (Proposition $\ref{prop:positivediffusion}$), we obtain straightforwardly the expected result for step $n+1$ with the lower bound $b_{n+1}(\tau)$.
\end{proof}
\bigskip


\subsubsection{A Maxwellian lower bound: proof of Theorem $\ref{theo:boundcutoff}$}\label{subsubsec:proofcutoff}

In this subsection we prove Theorem $\ref{theo:boundcutoff}$.

\bigskip
We take $f$ being the mild solution described in Theorem $\ref{theo:boundcutoff}$ and we suppose, thanks to Remarks $\ref{rem:tau0=0}$ and $\ref{rem:tau0=0diffusion}$, that $f$ satisfies Propositions $\ref{prop:upheaval}$ and $\ref{prop:positivediffusion}$ with $\tau_0=0$.
\par We fix $\tau>0$ and we keep the notations defined in Section $\ref{subsubsec:spreadingmethod}$ for the sequences $(r_n)$, $(a_n)$ and $(b_n)$ (see $\eqref{rn}$-$\eqref{an}$-$\eqref{bn}$, with $(\xi_n)$ to be defined later).

\bigskip
In Proposition $\ref{prop:spread}$, we showed that we can spread the localised lower bound with larger and larger velocities taken into account, i.e. by taking $\xi_n = 1/4$ the sequence $r_n$ is strictly increasing to infinity. We can consider $r_0>0$ and find an $n_0$ in $\N$ such that $B(0,R_{min})$
$$\forall v \in B(0,R_{min}),\quad B(0,r_0) \subset B(v,r_{n_0}).$$
By setting $N$ to be this specific $n_0$ and applying Proposition $\ref{prop:spread}$ with this $N$ we obtain a uniform lower bound:
$$\forall t \in \left[\tau-\frac{\tau}{2^{n_0+1}\tilde{r}_{n_0}},\tau\right],\:\forall x \in \bar{\Omega}, \quad f(t,x,v) \geq \min\left(a_{n_0}(\tau),b_{n_0}(\tau)\right) \mathbf{1}_{B(0,r_{n_0})}(v).$$
This bound is uniform in $x$ and the same arguments as in the proof of Proposition $\ref{prop:spread}$ allows us to spread it in the same manner in an even easier way since it is a global lower bound. Therefore, without loss of generality, we can assume $n_0 =0$ and that the following holds,
$$\forall n \in \N,\:\forall t \in \left[\tau-\frac{\tau}{2^{n+1}\tilde{r}_{n}},\tau\right],\:\forall x \in \bar{\Omega}, \quad f(t,x,v) \geq \min\left(a_n(\tau),b_n(\tau)\right) \mathbf{1}_{B(0,r_n)}(v),$$
with $(r_n)$, $(a_n)$ and $(b_n)$ satisfying the same inductive properties $\eqref{rn}$-$\eqref{an}$-$\eqref{bn}$ (with $r_{0} = r_{n_0}$, $a_0 = a_{n_0}$ and $b_0 = b_{n_o}$), with $(\xi_n)$ to be chosen later.

\bigskip
The proof of Theorem $\ref{theo:boundcutoff}$ is then done in two steps. The first one is to establish the Maxwellian lower bound at time $\tau$ using a slightly modified version of the argument in \cite{Mo2} Lemma $3.3$. The second is to prove that the latter bound holds for all $t>\tau$.

\bigskip
\textbf{$1^{st}$ step: a Maxwellian lower bound at time $\tau$.} A natural choice for $(\xi_n)$ is a geometric sequence $\xi_n = \xi^n$ for a given $\xi$ in $(0,1)$. With such a choice we have that
\begin{equation}\label{suprn}
r_n \leq r_02^{\frac{n}{2}}
\end{equation}
and
\begin{equation}\label{minrn}
r_n = 2^{\frac{n}{2}}r_0\prod\limits_{k=1}^n(1-\xi^k) \geq c_r 2^{\frac{n}{2}},
\end{equation}
with $c_r >0$ depending only on $r_0$ and $\xi$.

\bigskip
It follows that $f$ satisfies the following property
\begin{equation}\label{centredboundtau}
\forall n \in \N,\:\forall x \in \bar{\Omega},\:\forall v \in B(0,c_r2^{\frac{n}{2}}), \quad f(\tau,x,v) \geq c_n,
\end{equation}
with $c_n = \min(a_n,b_n)$.
\par It has been proven in \cite{Mo2} Lemma $3.3$ that for a function satisfying the property $\eqref{centredboundtau}$ with $c_n \geq \alpha^{2^n}$, for some $\alpha >0$, there exist $\rho$ and $\theta$ strictly positive explicit constants such that
$$\forall x \in \bar{\Omega},\:\forall v \in \R^d, \quad f(\tau,x,v) \geq \frac{\rho}{(2\pi \theta )^{d/2}}e^{-\frac{\abs{v}^2}{2\theta }}.$$
It thus only remains to show that there exists $\alpha_1$ and $\alpha_2$ strictly positive such that $b_n \geq \alpha_1^{2^n}$ and $a_n \geq \alpha_2^{2^n}$.

\bigskip
The case of $(b_n)$ is quite straightforward from $\eqref{suprn}$ and $0\leq \gamma^+ \leq 1$. Indeed, there exist an explicit constants $C_1$ and $C_2>0$, independent of $n$, such that for all $n\geq 1$
\begin{eqnarray*}
b_n = b_\partial e^{-C_L \tau \langle\tilde{r}_{n} \rangle^{\gamma^+}} \frac{e^{-\frac{\left(\tilde{r}_{n}\right)^2}{2T_{\partial}}}}{\left(2\pi\right)^{\frac{d-1}{2}}T_{\partial}^{\frac{d+1}{2}}} &\geq& \frac{b_\partial}{2T_{\partial}\left(2\pi\right)^{\frac{d-1}{2}}}e^{-C_1\left(R_{min}+r_{n}\right)^2}
\\&\geq& C_2 e^{-2C_1r_0^22^n}.
\end{eqnarray*}
Therefore if $C_2\geq 1$ we define $\alpha_1=\min(b_\partial,e^{-2C_1r_0^2})$ or else we define $\alpha_1=\min(b_\partial,C_2e^{-2C_1r_0^2})$ and it yields $b_n \geq \alpha_1^{2^n}$ for all $n\geq 0$.

\bigskip
We recall the inductive definition of $(a_n)$ for $n\geq 0$, with $\xi_n = \xi^n$,
$$a_{n+1}= \min\left(a_n,b_n\right)^2C_Qr_n^{d+\gamma}\xi^{(n+1)(d/2-1)}\frac{\tau}{2^{n+2}\tilde{r}_{n+1}}e^{-C_L \frac{\tau \langle \tilde{r}_{n+1} \rangle^{\gamma^+}}{2^{n+2}\tilde{r}_{n+1}}}.$$

First, using $\eqref{suprn}$ we have that for all $n\geq 0$,
\begin{eqnarray*}
-C_L \frac{\tau \langle \tilde{r}_{n+1} \rangle^{\gamma^+}}{2^{n+2}\tilde{r}_{n+1}} &\leq& \frac{C_L\tau R_{min}^{\gamma^+}}{2R_{min}} +  \frac{C_L\tau r_{n+1}^{\gamma^+}}{2^{n+2}R_{min}}
\\&\leq& \frac{C_L\tau R_{min}^{\gamma^+}}{2R_{min}} +  \frac{C_L\tau r_{0}^{\gamma^+}2^{\frac{(n+1)\gamma^+}{2}}}{2^{n+2}R_{min}},
\end{eqnarray*}
which is bounded from above for all $n$ since $0\leq \gamma^+ \leq 1$. Therefore, if we denote by $C_3$ any explicit non-negative constant independent of $n$, we have for all $n\geq 0$
$$a_{n+1} \geq C_3\frac{r_n^{d+\gamma}\xi^{((n+1))(d/2-1)}}{2^{n+2}(R_{min} + r_{n+1})}   \min\left(a_n,b_n\right)^2.$$
Thus using $\eqref{suprn}$ and $\eqref{minrn}$ to bound $r_n$ we have
\begin{eqnarray*}
a_{n+1} &\geq& C_3\frac{2^{\frac{n(d+\gamma)}{2}}\xi^{(n+1)(d/2-1)}}{2^{n+2}(R_{min} + c_r2^{\frac{n+1}{2}})}   \min\left(a_n,b_n\right)^2
\\&\geq& C_3\frac{2^{\frac{n(d+\gamma)}{2}}\xi^{n(d/2-1)}}{2^{\frac{3n+5}{2}}}   \min\left(a_n,b_n\right)^2
\\&\geq& C_3\left(\frac{2^{\frac{(d+\gamma)}{2}}\xi^{(d/2-1)}}{2^{\frac{3}{2}}} \right)^n  \min\left(a_n,b_n\right)^2.
\end{eqnarray*}
We define
$$\lambda = \min\left(1,C_3\right)\frac{2^{\frac{(d+\gamma)}{2}}\xi^{(d/2-1)}}{2^{\frac{3}{2}}},$$
which leads to,
\begin{equation}\label{an+1lowerbound}
\forall n \geq 1, \quad a_{n+1} \geq \lambda^n \min(a_n,b_n)^2.
\end{equation}

\bigskip
We could always have chosen $b_0$ and then $b_1$ respectively smaller than $a_0$ and $a_1$ (by always bounding from below by the minimum) and we assume that it is so. We can therefore define
$$\forall n \geq 1, \quad k_n = \min\left\{0\leq k \leq n-1:\quad a_{n-k} \geq b_{n-k}\right\}.$$
Notice that  $n-k_n \geq 1$, hence $\eqref{an+1lowerbound}$ can be iterated $k_n$ times which yields
\begin{eqnarray*}
\forall n \geq 1, \quad a_{n+1} &\geq& \lambda^{n+2(n-1) + \cdots + 2^{k_n}(n-k_n)}\min\left(a_{n-k_n},b_{n-k_n}\right)^{2^{k_n+1}}
\\ &\geq& \lambda^{2^{k_n +1}(n-k+1) - (n+2)}\left(\alpha_1^{2^{n-k_n}}\right)^{2^{k_n+1}}.
\end{eqnarray*}
Thus, if $\lambda \geq 1$ we can choose $\alpha_2 = \alpha_1$ and else we have $2^{k_n +1}(n-k+1) - (n+2)\leq 2^{n+1}$ and we can choose $\alpha_2 = \lambda\alpha_1$. In any case, $\alpha_2$ does not depend on $n$ and we have $a_n \geq \alpha_2^{2^n}$.

\bigskip
We therefore proved the existence of $\alpha >0$ such that for all $n$ in $\N$, $\min(a_n,b_n)\geq \alpha^{2^n}$ which implies that there exist $\rho$ and $\theta$ strictly positive explicit constants such that
\begin{equation}\label{boundcutofftau}
\forall x \in \bar{\Omega},\:\forall v \in \R^d, \quad f(\tau,x,v) \geq \frac{\rho}{(2\pi \theta )^{d/2}}e^{-\frac{\abs{v}^2}{2\theta }}.
\end{equation}

\bigskip
\textbf{$2^{nd}$ step: a Maxwellian lower bound for all $T > t\geq \tau$.}
To complete the proof of Theorem $\ref{theo:boundcutoff}$, it remains to prove that $\eqref{boundcutofftau}$ actually holds for all $t\geq\tau$. All the results and constants we obtained so far do not depend on an explicit form of $f_0$ but just on uniform bounds and continuity that are satisfied at all times, positions and velocities (by assumption). Therefore, we can do the same arguments starting at any time and not $t=0$. So if we take $\tau >0$ and consider $\tau \leq t < T$ we just have to make the proof start at  $t - \tau$ to obtain Theorem $\ref{theo:boundcutoff}$.
\bigskip


\subsection{A constructive approach to the initial lower bound and positivity of the diffussion}\label{subsec:constructiveupheavalpoint}

In previous subsections, we can see that explicit and constructive constants are obtained from given initial lower bounds and uniform positivity of the Maxwellian diffusion. Therefore, a constructive approach to the latter two will lead to completely explicit constants in the Maxwellian lower bound, depending only on \textit{a priori} bounds on the solution and the geometry of the domain. 

\bigskip
\textbf{Localised ``upheaval points".}
A few more assumptions on $f_0$ and $f$ suffice to obtain a completely constructive approach for the ``upheaval points". This method is based on a property of the iterated $Q^+$ operator discovered by Pulvirenty and Wennberg \cite{PW} and reformulated by Mouhot (\cite{Mo2} Lemma $2.3$) as follows.

\bigskip
\begin{lemma}\label{lem:Q+Q+}
Let $B=\Phi b$ be a collision kernel satisfying $\eqref{assumptionB}$, with $\Phi$ satisfying $\eqref{assumptionPhi}$ or $\eqref{assumptionPhimol}$ and $b$ satisfying $\eqref{assumptionb}$ with $\nu \leq 0$. Let $g(v)$ be a nonnegative function on $\R^d$ with bounded energy $e_g$ and entropy $h_g$ and a mass $\rho_g$ such that $0 < \rho_g < +\infty$. Then there exist $R_0 \:, r_0\:, \eta_0 >0$ and $\bar{v} \in B(0,R_0)$ such that
$$Q^+\left(Q^+\left(g\mathbf{1}_{B(0,R_0)},g\mathbf{1}_{B(0,R_0)}\right),g\mathbf{1}_{B(0,R_0)}\right) \geq \eta_0 \mathbf{1}_{B(\bar{v},r_0)},$$
with $R_0 \:, r_0\:, \eta_0$ being constructive in terms on $\rho_g$, $e_g$ and $h_g$.
\end{lemma}
\bigskip

We now suppose that $0<\rho_{f_0}<+\infty$, $h_{f_0}<+\infty$ and that
$$\forall (x,v) \in \Omega\times\R^d, \quad f_0(x,v) \geq \varphi(v) > 0$$
and we consider $R_0$, $r_0$, $\eta_0$ and $\bar{v}$ from Lemma $\ref{lem:Q+Q+}$ associated to the function $\varphi$.

\bigskip
We consider $x_1$ is in $\Omega$ and we denote $d_1 = d(x_1,\partial\Omega)$ the distance between $x_1$ and $\partial\Omega$. Define $\Delta_1 = \min(1,d_1/3R_0)$.
\par Take $0<\tau_0\leq \Delta_1$ and $v$ in $B(0,R_0)$.
\par We have, by construction that
\begin{equation}\label{construct0}
\forall t \in [0,\Delta_1], \quad \norm{x_1-(x_1-vt)}\leq \frac{d_1}{3},
\end{equation}
which means that $t_\partial(x,v) > \Delta_1$. By the Duhamel form without contact $\eqref{mildCO0}$ of $f$ and Lemma $\ref{lem:L}$ we have for all $t$ in $[\tau_0,\Delta_1]$,
\begin{equation}\label{construct1}
f(t,x_1,v) \geq f_0(x,v)e^{-tC_L\langle v \rangle^{\gamma^+}}\geq \varphi(v)e^{-tC_L\langle R_0 \rangle^{\gamma^+}}
\end{equation}
and
\begin{eqnarray}
f(t,x_1,v) &\geq&  \int_{0}^{t}e^{-(t-s) C_L \langle v \rangle^{\gamma^+}}  Q^+ \left[f(s,x_1-(t-s)v,\cdot),f(s,x_1-(t-s)v,\cdot)\right]\left(v\right)\:ds\nonumber
\\&\geq& e^{-(\Delta_1) C_L \langle R_0 \rangle^{\gamma^+}}\label{construct2}
\\&& \int_{0}^{t} Q^+ \left[f(s,x_1-(t-s)v,\cdot)\mathbf{1}_{B(0,r_0)},f(s,x_1-(t-s)v,\cdot)\mathbf{1}_{B(0,r_0)}\right]\left(v\right)\:ds.\nonumber
\end{eqnarray}

\bigskip
Now we notice that $\eqref{construct0}$ implies the following.
$$\forall s \in [0,t],\quad d(x_1-(t-s)v,\partial\Omega)\geq 2d_1/3.$$
Hence, if we call $v_*$ the integral variable in $Q^+$, we have that for all $v_*$ in $B(0,R_0)$ and all $s$ in $[0,t]$, $t_\partial(x-(t-s)v,v_*)>t$ (same arguments as for $\eqref{construct0}$). The function $f(s,x_1-(t-s)v,v_*)$ thus satisfies $\eqref{construct1}$ and $\eqref{construct2}$ as well. 
\par We can do this iteration one more time since $R_0\Delta_1 \leq d_1/3$ and this yields
\begin{eqnarray*}
f(t,x_1,v) \geq \tau_0^2 e^{-(3\Delta_1) C_L \langle R_0 \rangle^{\gamma^+}}
 \int_{0}^{t} Q^+\left(Q^+\left(\varphi\mathbf{1}_{B(0,R_0)},\varphi\mathbf{1}_{B(0,R_0)}\right),\varphi\mathbf{1}_{B(0,R_0)}\right)\left(v\right)\:ds.
\end{eqnarray*}
Applying Lemma $\ref{lem:Q+Q+}$ and remembering that $\Delta_1\leq 1$ we obtain that
\begin{equation}\label{constructboundx1}
\forall t \in [\tau_0,\Delta_1],\:\forall v \in B(0,R_0), \quad f(t,x_1,v) \geq \tau_0^3 e^{-3 C_L \langle R_0 \rangle^{\gamma^+}}\eta_0\mathbf{1}_{B(\bar{v},r_0)}.
\end{equation}

\bigskip
$f$ is uniformly continuous on the compact $[0,T/2] \times \bar{\Omega} \times B(0,R_0)$ so there exists $\delta_T$, $\delta_X$, $\delta_V>0$ such that
$$\forall |t-t'|\leq \delta_T, \: \forall \norm{x-x'}\leq \delta_X, \: \forall \norm{v-v'}\leq \delta_V,$$
\begin{equation}\label{uniformcontinuityconstruct}
\abs{f(t,x,v)- f(t',x',v')} \leq a_0(\tau_0),
\end{equation}
where we defined $2a_0(\tau_0) = \tau_0^3 e^{-3 C_L \langle R_0 \rangle^{\gamma^+}}\eta_0$.

\bigskip
From $\eqref{constructboundx1}$ and $\eqref{uniformcontinuityconstruct}$, we find
\begin{equation}\label{constructboundx1local}
\forall t \in [\tau_0,\Delta_1],\:\forall x \in B(x_1,\delta_X),\:\forall v \in B(0,R_0), \quad f(t,x,v) \geq \tau_0^3 e^{-3 C_L \langle R_0 \rangle^{\gamma^+}}\eta_0\mathbf{1}_{B(\bar{v},r_0)}.
\end{equation}

\bigskip
To conclude we construct $x_2,\dots,x_{N_X}$ such that $\bar{\Omega} \subset \bigcup\limits_{1 \leq i \leq N_X}B\left(x_i,\delta_X\right)$. We can use exactly the same arguments as for $x_1$ on each $x_i$ on the time interval $[\tau_0,\Delta_i]$ and we reach the conclusion $\eqref{constructboundx1local}$ on each $B(x_i,\delta_X)$. This gives, with $\Delta = \min(\Delta_i)$,
\begin{equation}\label{constructupheavalfinal}
\forall t \in [\tau_0,\Delta],\:\forall x\in \bar{\Omega},\:\forall v \in B(0,R_0), \quad f(t,x,v) \geq a_0(\tau_0)\mathbf{1}_{B(\bar{v},r_0)}.
\end{equation}

\bigskip
\begin{remark}
We emphasize here that even though we used compactness arguments, they appeared to be solely a technical artifice. The constants $a_0(\tau_0)$, $\bar{v}$, $R_0$ and $r_0$ are entirely explicit and depends on \textit{a priori} bounds on $f$. The only point of concern would be that $\Delta$ is not constructive since it depends on the covering. However, in previous sections, only the constant bounding the diffusive process (Proposition $\ref{prop:positivediffusion}$ contains a dependency in $\Delta$ (see $\eqref{positivitydiffusion1}$ and $\eqref{finalboundiffusion}$) but depend on an upper bound on $\Delta$ which is less than $1$.
\end{remark}
\bigskip

Starting from this explicit bound we can use the proofs made in the previous subsections, that are constructive, to therefore have a completely constructive proof as long as the bound on the diffusion (Proposition $\ref{prop:positivediffusion}$) can be obtained without compactness arguments.

\bigskip
\textbf{Constructive approach of positivity of diffusion.}
A quick look at the proof of Proposition $\ref{prop:positivediffusion}$ shows that we only need to construct $\delta(\tau_0)$ in $\eqref{positivitydiffusion1}$ and $\delta'(x_1)$ in $\eqref{positivitydiffusion2}$ explicitly.

\bigskip
The first equality $\eqref{positivitydiffusion1}$ is obtained constructively combining the arguments to obtain it pointwise (see proof of Proposition $\ref{prop:positivediffusion}$) together with the method of the proof of Lemma $\ref{lem:positivity1}$. 
\par Indeed, take $x_\partial$ on $\partial\Omega$ and fix an $x_0$ in $\omega$. We can grow the initial lower bound $\eqref{constructupheavalfinal}$ at $x_1$ for $t$ in $[\tau_0,\Delta]$ such that it includes $B(0,2d_\Omega/\tau_0 + \norm{\bar{v}})$ (as in Lemma $\ref{lem:positivity1}$). Then, as in the beginning of the proof of Proposition $\ref{prop:upheaval}$ we obtain that, 
\begin{equation}\label{constructdiffusion1}
\forall t \in [\tau_0,\Delta], \quad f\left(t,x_\partial,2\frac{x_\partial-x_0}{\tau_0}\right) \geq A>0.
\end{equation}
We can now do that for all $v$ in 
$$B\left(2\frac{x_\partial-x_0}{\tau_0},\min\left(r_0;2\frac{d(x_0,\partial\Omega)}{\tau_0}\right)\right)$$
 by just defining, for each $v$, $x(v)$ the point in $\Omega$ such that
$$v = 2\frac{x_\partial - x(v)}{\tau_0}.$$
Note that this point is always well defined since $\Omega$ is convex and
$$\norm{x_0-x(v)}\leq d(x_0,\partial\Omega).$$
For any given $x(v)$ we apply the same argument as for $x_0$ so that the lower bound includes $v$. Therefore, there is an infimum constant $A_{min}$ satisfying, 
$$\forall t \in [\tau_0,\Delta_{min}],\forall v \in B\left(2\frac{x_\partial-x_0}{\tau_0},\min\left(r_0;2\frac{d(x_0,\partial\Omega)}{\tau_0}\right)\right), \quad f\left(t,x_\partial,v\right) \geq A_{min}>0,$$
which is a constructive version of $\eqref{positivitydiffusion1}$. We emphasize here that $A_{min}$ exists and is indeed independent on $v$ since it depends on the number of iteration of Lemma $\ref{lem:positivity1}$ (itself determined by the norm of $v$ which is bounded) and the initial lower bound at $x(v)$ which is uniform in space by $\eqref{constructupheavalfinal}$.

\bigskip
The second inequality $\eqref{positivitydiffusion2}$ is purely geometric as long as we fix a $x_1$ satisfying the initial lower bound, which is the case with $x_0$ used above. We therefore obtained an entirely constructive method for the positivity of the diffusion process.
\bigskip

\section{The non-cutoff case: an exponential lower bound} \label{sec:noncutoff}

In this section we prove the immediate appearance of an exponential lower bound for solutions to the Boltzmann equation $\eqref{BE}$ in the case of a collision kernel satisfying the non-cutoff property.
\par The definition of being a mild solution in the case of a non-cutoff collision kernel, Definition $\ref{def:mildnoncutoff}$ and equation $\eqref{noncutoffsplitting}$, shows that we are in fact dealing with a cutoff kernel to which we add a non locally integrable remainder. As we shall see in Section $\ref{subsec:controlSQ1}$, $S_\eps$ enjoys the same kind of $L^\infty$ control than the operator $L$ whilst $Q^1_\eps$, the non-cutoff part of the gain operator, has an $L^\infty$-norm that decreases to zero as $\eps$ goes to zero.
\par The strategy used is therefore utterly identical to the cutoff case: creation of localised ``upheaval points'' and spreading of these initial lower bounds up to an exponential lower bound. The difference will be that, at each step $n$ of the spreading process we will choose an $\eps_n$ and a $\xi_n$ such that the perturbation added by the non-cutoff part $-\norm{Q^1_{\eps_n}}_{L^\infty_v}$ still preserves a uniform positivity in larger and larger balls in velocity.
\par Note that the uniform positivity of the Maxwellian diffusion still holds in the non-cutoff case since it only comes from an initial positivity and the geometry of the domain.
\bigskip


\subsection{Controls on the operators $S_\eps$ and $Q^1_\eps$}\label{subsec:controlSQ1}

We gather here two lemmas, proven in \cite{Mo2}, which we shall use in this section. They control the $L^\infty$-norm of the linear operator $S_\eps$ and of the bilinear operator $Q^1_\eps$. We first give a property satisfied by the linear operator $S$, $\eqref{noncutoffsplitting}$,  which is Corollary $2.2$ in \cite{Mo2}, where we define
\begin{equation}\label{mb}
m_b = \int_{\mathbb{S}^{d-1}}b\left(\mbox{cos}\:\theta\right)(1-\mbox{cos}\:\theta)d\sigma = \left|\mathbb{S}^{d-2}\right|\int_0^\pi b\left(\mbox{cos}\:\theta\right)(1-\mbox{cos}\:\theta) \mbox{sin}^{d-2}\theta \:d\theta.
\end{equation}

\bigskip
\begin{lemma}\label{lem:S}
Let $g$ be a measurable function on $\R^d$. Then
$$\forall v \in \R^d,\quad \abs{S[g](v)} \leq C_g^S\langle v \rangle^{\gamma^+},$$
where $C_g^S$ is defined by:
\begin{enumerate}
\item If $\Phi$ satisfies $\eqref{assumptionPhi}$ with $\gamma \geq 0$ or if $\Phi$ satisfies $\eqref{assumptionPhimol}$, then
$$C^S_g =\emph{\mbox{cst}}\: m_b C_\Phi e_g.$$
\item If $\Phi$ satisfies $\eqref{assumptionPhi}$ with $\gamma \in (-d,0)$, then
$$C^S_g = \emph{\mbox{cst}}\: m_b C_\Phi \left[e_g+ l^p_g\right],\quad p > d/(d+\gamma).$$
\end{enumerate}
\end{lemma}
\bigskip

We will compare the lower bound created by the cutoff part of our kernel to the remaining part $Q^1_\eps$. To do so we need to control its $L^\infty$-norm. This is achieved thanks to Lemma $2.5$ in \cite{Mo2}, which we recall here.

\bigskip
\begin{lemma}\label{lem:Q1}
Let $B=\Phi b$ be a collision kernel satisfying $\eqref{assumptionB}$, with $\Phi$ satisfying $\eqref{assumptionPhi}$ or $\eqref{assumptionPhimol}$ and $b$ satisfying $\eqref{assumptionb}$ with $\nu \in [0,2)$. Let $f,g$ be measurable functions on $\R^d$.
\\Then
\begin{enumerate}
\item If $\Phi$ satisfies $\eqref{assumptionPhi}$ with $2+\gamma \geq 0$ or if $\Phi$ satisfies $\eqref{assumptionPhimol}$, then
$$\forall v \in \R^d,\quad \abs{Q^1_b(g,f)(v)} \leq \emph{\mbox{cst}}\: m_b C_\Phi \norm{g}_{L^1_{\tilde{\gamma}}}\norm{f}_{W^{2,\infty}}\langle v \rangle ^{\tilde{\gamma}}.$$
\item If $\Phi$ satisfies $\eqref{assumptionPhi}$ with $2+\gamma < 0$, then
$$\forall v \in \R^d,\quad \abs{Q^1_b(g,f)(v)} \leq \emph{\mbox{cst}}\: m_b C_\Phi \left[\norm{g}_{L^1_{\tilde{\gamma}}}+\norm{g}_{L^p}\right]\norm{f}_{W^{2,\infty}}\langle v \rangle ^{\tilde{\gamma}}$$
with $p > d/(d+\gamma+2)$.
\end{enumerate}
\end{lemma}
\bigskip


\subsection{Proof of Theorem $\ref{theo:boundnoncutoff}$}\label{subsec:proofNCO}

As explained at the beginning of this section, the main idea is to compare the loss due to the non-cutoff part of the operator $Q^1_\eps$ with the spreading properties of the cutoff operator $Q^+_\eps$. More precisely, due to Lemmas $\ref{lem:L}$, $\ref{lem:S}$ and $\ref{lem:Q1}$ we find that for all $0<\eps<\eps_0$,
\begin{equation}\label{L+S}
L_{\eps}[f] +  S_{\eps}[f] \leq C_f \left(n_{b^{CO}_{\eps}}+ m_{b^{NCO}_{\eps}}\right)\langle v\rangle^{\gamma^+}
\end{equation}
and
\begin{equation}\label{Q+Q1}
Q^+_\eps(f,f) + Q^1_\eps(f,f) \geq Q^+_\eps(f,f) -\abs{ Q^1_\eps(f,f)} \geq  Q^+_\eps(f,f) - C_f m_{b^{NCO}_{\eps}}\langle v\rangle^{(2 + \gamma)^+},
\end{equation}
where $C_f >0$ is a constant depending on $E_f$, $E'_f$, $W_f$ (and $L^{p_\gamma}_f$ if  $\Phi$ satisfies $\eqref{assumptionPhi}$ with $\gamma < 0$).
\par Moreover, by definitions $\eqref{lb}$, $\eqref{nb}$ and $\eqref{mb}$, the following behaviours happen:
$$l_{b^{CO}_{\eps}} \geq l_b$$
and
\begin{equation}\label{eqv1}
n_{b^{CO}_{\eps}} \mathop{\sim}\limits_{\eps \to 0} \frac{b_0}{\nu}\eps^{-\nu}, \quad m_{b^{NCO}_{\eps}} \mathop{\sim}\limits_{\eps \to 0} \frac{b_0}{2-\nu}\eps^{2-\nu}
\end{equation}
if $\nu$ belongs to $(0,2)$ and
\begin{equation}\label{eqv2}
n_{b^{CO}_{\eps}} \mathop{\sim}\limits_{\eps \to 0} b_0\abs{\mbox{log}\eps}, \quad m_{b^{NCO}_{\eps}} \mathop{\sim}\limits_{\eps \to 0} \frac{b_0}{2}\eps^{2}
\end{equation}
for $\nu=0$. This shows that the contribution of $Q^1_\eps$ decreases with $\eps$ so this operator should not affect the spreading method whereas the contribution of $S_\eps$ increases, which is why we lose the Maxwellian lower bound to get a faster exponential one.
\par We just briefly describe the changes to make into the proof of the cutoff case to obtain Theorem $\ref{theo:boundnoncutoff}$.

\bigskip
\textbf{Localised ``upheaval points''.}
The creation of localised initial lower bounds (Proposition $\ref{prop:upheaval}$ in the cutoff case) depends on the Boltzmann operator for two different reasons:
\begin{itemize}
\item the creation of a first lower bound in a neighborhood of a point $(x_1,v_1)$ in the phase space (Lemma $\ref{lem:positivity1}$)
\item the creation of localised lower bounds in $\bar{\Omega}$ \textit{via} the free transport part.
\end{itemize}
Since $L_\eps +S_\eps$ satisfies the same bounds as $L$ in the cutoff case, the second step can be made identically as in the proof of Proposition $\ref{prop:upheaval}$. It remains to prove the creation of an initial bound in a neighborhood of $(x_1,v_1)$.

\bigskip
We use the same definition of $\Delta$, $x_1$, $v_1$, $\alpha_n(t)$ and $(r_n)_{n\in\N}$ as in Lemma $\ref{lem:positivity1}$ apart from 
$$t_n(t)= \max\left\{\tau_0, t-\frac{\Delta}{2^{n+1}\left(\norm{v_1}+r_n\right)}\right\}.$$
Note that the step $n=0$ holds here since it only depends on the continuity of the solution $f$ and $f_0$.
\par The key difference will be that the equivalent statement of Lemma $\ref{lem:positivity1}$ can only be done on time intervals of the form $[\tau_0,\Delta]$ for any $\tau_0>0$. Indeed, take $\tau_0>0$ and therefore
$$\forall n\in \N,\:\forall t \in [\tau_0,\Delta],\quad \alpha_n(t)\geq \alpha_n(\tau_0).$$
Exactly the same arguments as in the inductive step $n+1$ in the proof of Lemma $\ref{lem:positivity1}$ we reach $\eqref{inductionpositivity}$ with our new operators (with a cutoff function depending on $\eps_{n+1}$)
\begin{eqnarray*}
f(t,x,v) &\geq& \int_{t_n(t)}^{t}e^{-sC_f^{\eps_{n+1}}(R)}
\\&&\quad\quad\quad\left(\alpha_n(\tau_0)^2 Q^+_{\eps_{n+1}}[\mathbf{1}_{B(\bar{v},r_n)},\mathbf{1}_{B(\bar{v},r_n)}] - C_f m_{b^{NCO}_{\eps_{n+1}}}\langle R\rangle^{(2 + \gamma)^+}\right)(v)\: ds,
\end{eqnarray*}
using the shorthand notations $C^\eps_f(R) = C_f(n_{b^{CO}_{\eps}}+m_{b^{NCO}_{\eps}})\langle R\rangle^{\gamma^+}$ and $R =\norm{v_1}+ 2r_n$.
Due to the spreading property of $Q^+_{\eps_{n+1}}$ (see Lemma $\ref{lem:Q+spread}$) with $\xi=1/4$ we reach
\begin{eqnarray}
&& \quad f(t,x,v) \geq \int_{t_n(t)}^{t}e^{-C_f^{\eps_{n+1}}(R)} \label{finalineqNCO} 
\\ &&\quad\quad\left(\alpha_n^2(\tau_0)\mbox{cst}\: l_{b^{CO}_{\eps_{n+1}}} c_\Phi r_n^{d+\gamma} \xi^{\frac{d}{2}-1}\mathbf{1}_{B\left(\bar{v},r_n\sqrt{2}(1-\xi)\right)}- C_f m_{b^{NCO}_{\eps_{n+1}}}\langle R\rangle^{(2 + \gamma)^+}\right)(v)\: ds,\nonumber
\end{eqnarray}
Thus, at each step of the induction we just have to choose $\eps_{n+1}$ small enough such that
\begin{equation}\label{choiceepsn}
C_f m_{b^{NCO}_{\eps_{n+1}}}\langle R\rangle^{(2 + \gamma)^+} \leq \frac{1}{2} \alpha_n^2(\tau_0)\mbox{cst}\: l_b c_\Phi r_n^{d+\gamma} \xi^{\frac{d}{2}-1}.
\end{equation}
This proves the following Lemma.

\bigskip
\begin{lemma}\label{lem:positivityNCO}
Let $f$ be the mild solution of the Boltzmann equation described in Theorem $\ref{theo:boundnoncutoff}$.
\\ Then there exist $\Delta >0$, $(x_1,v_1)$ in $\Omega \times \R^d$ such that for all $n \in \N$ there exist $r_n >0$, depending only on $n$, such that for all $\tau_0$ in $(0,\Delta]$ there exists and $\alpha_n(\tau_0)>0$ such that for
$$\forall t \in [\tau_0,\Delta],\:\forall x \in B\left(x_1,\frac{\Delta}{2^n}\right), \:\forall v \in \R^d, \quad f(t,x,v) \geq \alpha_n(\tau_0)\mathbf{1}_{B(v_1,r_n)}(v).$$
\end{lemma}

\bigskip
\textbf{Exponential lower bound.}
As explained before, the strict positivity of the diffusion still holds in our case since we proved the initial lower bound in Lemma $\ref{lem:positivityNCO}$. It therefore remains to show that we can indeed spread the ``upheaval points''. This is achieved by adapting the arguments of the cutoff case together with careful choices of $\eps_{n+1}$ and $\xi_{n+1}$ at each step of the induction. This has been done in \cite{Mo2} and in \cite{Bri2} and we refer to these works for deeper details.

\bigskip
Basically, we start by spreading the initial ''upheaval points" (obtained from Lemma $\ref{lem:positivityNCO}$ with the same method as Proposition $\ref{prop:upheaval}$) by induction. At each step of the induction we use the spreading property of the $Q^+_{\eps_n}$ operator between $t^{(2)}_n$ and $t^{(1)}_n$ (see $\eqref{inductionCO}$) and we fix $\eps_n$ small enough to obtain a strictly positive lower bound (see $\eqref{choiceepsn}$).
\par There is, however, a subtlety in the non-cutoff case that we have to deal with. Indeed, at each step of the induction we choose an $\eps_n$ of decreasing magnitude, but at the same time in each step the action of the operator $-(L_{\eps_n} + S_{\eps_n})$ behaves like (see $\eqref{finalineqNCO}$)
$$\mbox{exp}\left[-C_f\left(m_{b^{NCO}_{\eps_n}}+n_{b^{CO}_{\eps_n}}\right)(t^{(1)}_{n}-t^{(2)}_n)\langle v \rangle^{\gamma^+}\right].$$
By $\eqref{eqv1}-\eqref{eqv2}$, as $\eps_n$ tends to $0$ we have that $n_{b^{CO}_{\eps_n}}$ goes to $+\infty$ and so the action of $-(Q^-_\eps + Q^2_\eps)$ seems to decrease the lower bound to $0$ exponentially fast. The idea to overcome this difficulty is to find a time interval $t^{(1)}_{n}-t^{(2)}_n = \Delta_n$, at each step to be sufficiently small to counterbalance the effect of $n_{b^{CO}_{\eps_n}}$.
\par More precisely, taking $$t^{(1)}_{n} = \left(\sum_{k=0}^{n+1}\Delta_k\right)\tau, \quad t^{(2)}_{n} = \left(\sum_{k=0}^n\Delta_k\right)\tau\quad\mbox{with}\quad \sum\limits_{k=0}^\infty \Delta_k =1,$$
fixing $\eps_n$ by $\eqref{choiceepsn}$ and choosing carefully $\xi_n$ (exactly as in \cite{Mo2}\cite{Bri2})we reach the desired Theorem $\ref{theo:boundnoncutoff}$.
\bigskip




\bibliographystyle{acm}
\bibliography{bibliography}

\begin{thebibliography}{10}

\bibitem{Bri2}
{\sc Briant, M.}
\newblock Instantaneous filling of the vacuum for the full {B}oltzmann equation
  in bounded domains.
\newblock Preprint.

\bibitem{Ca}
{\sc Carleman, T.}
\newblock Sur la th\'eorie de l'\'equation int\'egrodiff\'erentielle de
  {B}oltzmann.
\newblock {\em Acta Math. 60}, 1 (1933), 91--146.

\bibitem{Ce}
{\sc Cercignani, C.}
\newblock {\em The {B}oltzmann equation and its applications}, vol.~67 of {\em
  Applied Mathematical Sciences}.
\newblock Springer-Verlag, New York, 1988.

\bibitem{Ce1}
{\sc Cercignani, C., Illner, R., and Pulvirenti, M.}
\newblock {\em The mathematical theory of dilute gases}, vol.~106 of {\em
  Applied Mathematical Sciences}.
\newblock Springer-Verlag, New York, 1994.

\bibitem{DV}
{\sc Desvillettes, L., and Villani, C.}
\newblock On the trend to global equilibrium in spatially inhomogeneous
  entropy-dissipating systems: the linear {F}okker-{P}lanck equation.
\newblock {\em Comm. Pure Appl. Math. 54}, 1 (2001), 1--42.

\bibitem{DV1}
{\sc Desvillettes, L., and Villani, C.}
\newblock On the trend to global equilibrium for spatially inhomogeneous
  kinetic systems: the {B}oltzmann equation.
\newblock {\em Invent. Math. 159}, 2 (2005), 245--316.

\bibitem{Gr1}
{\sc Grad, H.}
\newblock Principles of the kinetic theory of gases.
\newblock In {\em Handbuch der {P}hysik (herausgegeben von {S}. {F}l\"ugge),
  {B}d. 12, {T}hermodynamik der {G}ase}. Springer-Verlag, Berlin, 1958,
  pp.~205--294.

\bibitem{GMM}
{\sc Gualdani, M.~P., Mischler, S., and Mouhot, C.}
\newblock Factorization for non-symmetric operators and exponential
  {H}-theorem.

\bibitem{Gu6}
{\sc Guo, Y.}
\newblock Decay and continuity of the {B}oltzmann equation in bounded domains.
\newblock {\em Arch. Ration. Mech. Anal. 197}, 3 (2010), 713--809.

\bibitem{Mo2}
{\sc Mouhot, C.}
\newblock Quantitative lower bounds for the full {B}oltzmann equation. {I}.
  {P}eriodic boundary conditions.
\newblock {\em Comm. Partial Differential Equations 30}, 4-6 (2005), 881--917.

\bibitem{PW}
{\sc Pulvirenti, A., and Wennberg, B.}
\newblock A {M}axwellian lower bound for solutions to the {B}oltzmann equation.
\newblock {\em Comm. Math. Phys. 183}, 1 (1997), 145--160.

\bibitem{Vi2}
{\sc Villani, C.}
\newblock A review of mathematical topics in collisional kinetic theory.
\newblock In {\em Handbook of mathematical fluid dynamics, {V}ol. {I}}.
  North-Holland, Amsterdam, 2002, pp.~71--305.

\end{thebibliography}


\signmb
\end{document}